\documentclass{aims}

\usepackage{subcaption}
\usepackage{graphicx}
\usepackage{caption}

\usepackage{bbm}

\usepackage{txfonts}
\def\typeofarticle{Research Article}
\def\currentvolume{19}
\def\currentissue{x}
\def\currentyear{2024}

\def\ppages{xxx--xxx}
\def\DOI{}
\def\Received{December 2024}
\def\Revised{x}
\def\Accepted{x}
\def\Published{x}

\newcommand{\given}{\, | \,}
\newcommand{\with}{\, ; \,}
\renewcommand{\d}{\mathrm{d}}
\newcommand{\micron}{\mu\mathrm{m}}

\newcommand{\msd}{\mathcal{M}}

\newcommand{\csahat}{\widehat{\Psi}}

\newcommand{\Prev}{P_\mathrm{reverse}}
\newcommand{\Pcont}{P_\mathrm{continue}}
\newcommand{\sigmac}{\sigma_{\mathrm{cargo}}}

\newtheorem{theorem}{Theorem}[section]

\theoremstyle{remark}

\numberwithin{equation}{section}

\begin{document}

\title{Considering experimental frame rates and robust segmentation analysis of piecewise-linear microparticle trajectories}

\author{%
  Keisha J. Cook\affil{1}\corrauth,
  Nathan Rayens\affil{3},
  Linh Do\affil{2},
  Christine K.~Payne\affil{3},
  and
  Scott A.~McKinley\affil{2}
}

% \shortauthors is used in copyright information in the end of the paper
\shortauthors{the Author(s)}

\address{%
  \addr{\affilnum{1}}{School of Mathematical and Statistical Sciences, Clemson University, Clemson, SC 29634, USA}
  \addr{\affilnum{2}}{Department of Mathematics, Tulane University, New Orleans, LA 70118, USA}
  \addr{\affilnum{3}}{Thomas Lord Department of Mechanical Engineering and Materials Science, Duke University, Durham, NC 27708, USA}
  }

% corresponding author
\corraddr{keisha@clemson.edu; Tel: +1-864-656-3434.
}

\begin{abstract}
The movement of intracellular cargo transported by molecular motors is commonly marked by switches between directed motion and stationary pauses. The predominant measure for assessing movement is effective diffusivity, which predicts the mean-squared displacement of particles over long time scales. In this work, we consider an alternative analysis regime that focuses on shorter time scales and relies on automated segmentation of paths. Due to intrinsic uncertainty in changepoint analysis, we highlight the importance of statistical summaries that are robust with respect to the performance of segmentation algorithms. In contrast to effective diffusivity, which averages over multiple behaviors, we emphasize tools that highlight the different motor-cargo states, with an eye toward identifying biophysical mechanisms that determine emergent whole-cell transport properties. By developing a Markov chain model for noisy, continuous, piecewise-linear microparticle movement, and associated mathematical analysis, we provide insight into a common question posed by experimentalists: how does the choice of observational frame rate affect what is inferred about transport properties?

\end{abstract}

\keywords{
\textbf{Microparticle, single particle tracking, intracellular transport, changepoint}
}

\maketitle

\section{Introduction}

Intracellular transport has been studied widely to understand the movement and emergent organization of organelles and vesicles within biological cells \cite{mogre2020getting,Trackmate,Neumann,BARLAN2013483,Grafstein,GUZIK2004443,Staff20111015}. A variety of single particle tracking (SPT) techniques have been developed to capture individual microparticle trajectories as they evolve over time, both \emph{in vitro} and \emph{in vivo} \cite{BressloffNewby,SUH200563}, using methods that are shaped by a balance between emerging technological capabilities and the biological constraints of particular systems of interest. These methods have produced and will continue to produce fundamental insights into how micro- and nano-scale interactions accrue to produce whole-cell phenomena \cite{BA20183591,LeeEtAl,CAVISTON2006530,Lakadamyali2013,Klumpp2005,Klumpp2008,MULLER20102610,JONGSMA2016152,BARLAN2013483,KARKI199945,vale1987,Vale200088}. 

For intracellular transport mediated by processive molecular motors, the motion of biomolecular cargo is marked by switches between \emph{stationary} and \emph{motile} states. Time spent in either state can be affected by the propensity of motor proteins to either bind or unbind from microtubules, and their ability to process along microtubules when bound. The switches in behavior occur on a time scale of seconds, but transport of individual cargoes across a cell may require minutes, hours, or days. As a result, the temporal resolution at which location observations are made, or simply the \emph{frame rate}, has a profound impact on how trajectories are perceived. If a motor-cargo complex switches states on the order of once per second, then it is necessary to capture images at least 10 to 20 times per second to resolve the different behaviors. On the other hand, if the frame rate is on the order of once per 10 seconds or more, then switching is not evident and paths can look effectively Brownian.

Essential to these investigations are an array of mathematical models that link the transient and rapidly changing micro-scale interactions to stable features at the whole-cell scale. Two of the most popular ways to learn the parameters of these mathematical models are through \emph{segmentation analysis} and \emph{mean-squared displacement (MSD)}. Segmentation analysis is typically used for trajectories observed at ``fast'' frame rates where the hope is that the multitude of motor-cargo (or microparticle) behaviors are highly resolved and can be distinguished reliably \cite{JensenEtAl,YIN2018217,KARSLAKE202116,Chen2021,Monnier2015,Persson2013ExtractingID,Das2009,BarryHartigan,Erdman20071}. On the other hand, MSD analysis dispenses with the interest in fine-scale behavior and seeks to fit parameters that produce patterns observed at larger time scales \cite{Saxton,Manzo2015,Shen2017,Sokolov2012,QIAN1991910,RUTHARDT20111199,Michalet2010,Kepten2015,MONNIER2012616,Monnier2015}. However, both modes of investigation are prone to different types of systematic error. For MSD analysis, the inferred MSD curves are, by definition, averaging over multiple behaviors. It can happen that multiple fine-scale stochastic models can produce similar emergent MSD profiles. For segmentation analysis, the segmentation step itself is highly vulnerable to bias due to the choice of algorithm, or collection of humans, that do the work. 

In the work that follows, we focus on segmentation analysis, with an emphasis on identifying the impact of frame rate and the segmentation algorithm on different statistical summaries. Through analysis of simulated and experimental data, we demonstrate a statistical resilience that exists in identifying the proportion of time spent in different states. However, we highlight the limits of the method through a bias-variance tradeoff that emerges across different frame rates. As a consequence we show that, in terms of inference reliability, there is a real difference between two different notions of the term ``speed distribution'' and identify what appears to be a reliable metric for faithfully characterizing and comparing distinct populations of noisy, continuous, piecewise-linear trajectories.

\subsection{Background and experimental context}

In a typical SPT experiment, the microscope camera captures the positions of a particle over time at a defined number of frames per second (frame rate), which are combined to form a trajectory. Traditionally the frame rate has been decided by the size of target particles, biophysical constraints of the experimental system, the microscope's focal depth, and limitations on the duration of observation windows \cite{Shen2017,Manzo2015,Monnier2015,SUH200563}. It is natural to assume that experiments should be run at the fastest possible frame rates, but important tradeoffs exist.

The primary tradeoff we will study here arises when there are a limited number of observations that can be made of a single path. Accurate segmentation analysis requires fine enough temporal resolution that state transitions are identifiable, but if trajectories are too short, then too few transitions will be captured and faithful characterization will not be possible. One of the most common techniques in SPT, fluorescence microscopy, is fundamentally limited in this way due to \emph{photobleaching} \cite{10.1242/jcs.142943,Diekmann,LeleketAl}. Heuristically, particles of interest are labeled with fluorophores that emit photons when excited by a light source. This light exposure ultimately damages the fluorophores preventing fluorescence.  In terms of frame rate, if an experimentalist increases the frequency of observation to resolve finer time scales, the overall lives of fluorophores are shortened, leading to shorter trajectories. Experimentalists sometimes describe this as a "photon limited" process.

A second example of a frame-rate/path-length tradeoff exists in the segmentation algorithms themselves. The complexity of the task in identifying changepoints can grow with the square of the path length, and without proper penalization, substantial numbers of false positives emerge~\cite{CPOP}. Moreover, as more biophysical detail becomes visible, it can become a significant challenge to identify relevant state changes without chasing nuisance phenomena that emerge at the fastest frame rates. This was highlighted in work by Feng et al.~\cite{FengEtAl} in which the authors tracked individual molecular motor heads using gold nanoparticles at 1000 frames per second, see Fig.~3 in particular. The data is illuminating for studying individual motor steps and the propensity of motor heads to bind to microtubules when diffusing locally, but state-switches between stationary and motile phases are rarely observed at this time scale and difficult to identify. (See also, \cite{Mickolajczyk418053}.)

But experiments conducted with slower frame rates have their own tracking issues. Many imaging techniques are intrinsically two-dimensional in their observation and rendering. So, when particles move in the $z$-direction, they move in and out of the microscope's focal plane. This results in incomplete or broken trajectories, and cases where multiple individual particle trajectories are incorrectly combined. This phenomenon was considered explicitly by Wang et al.~\cite{wang2015minimizing} when attempting to conduct a ``census'' of fast and slow diffusing particles. The authors showed that a na\"ive count of paths leads to an overcount of fast-moving particles in the system. In another study, Lee et al.~\cite{LeeEtAl} considered the impact of frame rate on capturing molecules diffusing on a membrane, exhibiting switches between fast, intermediate, and immobile states. The authors opted for ``intermediate'' choices of 12 ms/frame ($\sim 83$ Hz) and 35 ms/frame ($\sim 28.5$ Hz).  Interestingly, the authors found that the signal-to-noise ratio was lower at 83 Hz, but molecules were lost by the tracker more frequently, especially in high-density regions. This led to a systemic challenge in estimating effective diffusivity. Some of these issues were addressed by Hansen et al.~\cite{SpotOnSPT} in the development of their SpotOn tool. Instead of using MSD techniques that rely on long particle trajectories, the authors developed an inference framework that pools together information from many short trajectories. They were able to robustly infer effective diffusivity across a range of fast frame rates (50-200 Hz) by explicitly accounting for particles moving in and out of the focal plane, motion-blur as a result of fast-diffusing molecules emitting photons over multiple pixels, tracking errors due to high particle densities, and under-counting of particles due to photobleaching. 

In the work we present here, the motivating data sets capture the movement of lysosomes in different types of cells. In \cite{Payne2014}, lysosomes were captured at a frame rate of 0.3s/frame (3.33 Hz) when studying the impact of diameter on intracellular transport. The authors estimated effective diffusivity from the slope of ensemble-averaged (MSD) curves and showed that increased lysosome diameter led to decreased effective diffusivity. However, the properties of active transport were not affected by lysosome diameter. This raises the challenge of developing a protocol that can identify whether reduced effective diffusivity is due to shorter periods of motility, or if motile durations were similar in length but relatively fewer compared to periods of stationarity. In a different experimental setting, we addressed this question when considering lysosomal transport in different regions of the cell \cite{RayensEtAl}. Using a faster frame rate (20 Hz) and employing Bayesian inference through a Markov chain model for state-switching, we were able to show that a regional difference in effective diffusivity was due to longer periods of stationarity (and not a difference in transport speeds or run lengths). This later led to a novel observation of an interaction between lysosome transport and the endoplasmic reticulum of cells \cite{RayensEtAl2}. 

\begin{figure}[H]
    \centering
    \begin{subfigure}[t]{0.48\textwidth}
        \centering
        \includegraphics[width=\textwidth]{{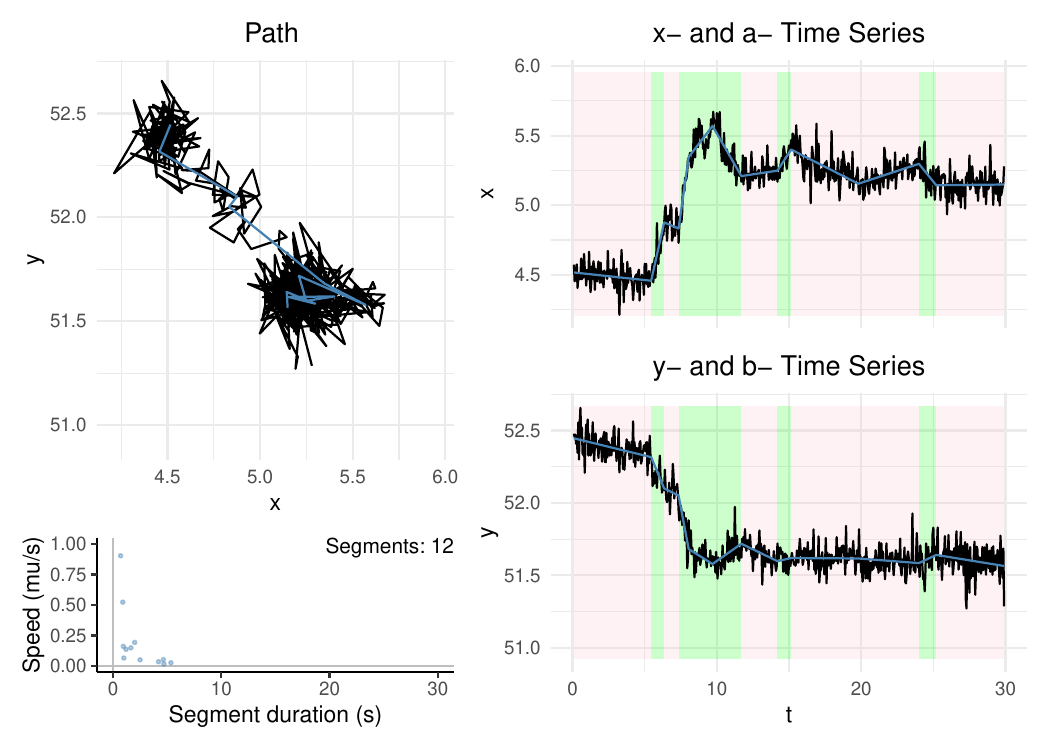}}
        \caption{One experimental lysosome trajectory in the periphery of a cell from \cite{RayensEtAl}.}
    \end{subfigure}
    ~
    \begin{subfigure}[t]{0.48\textwidth}
        \centering
        \includegraphics[width=\textwidth]{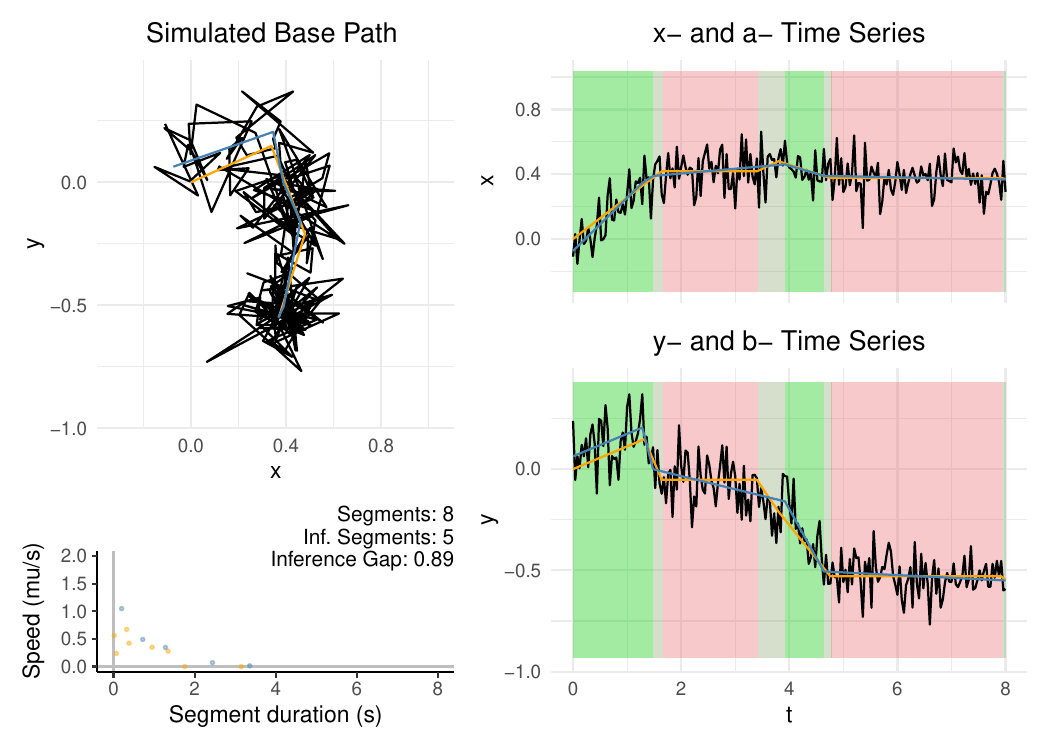}
        \caption{A trajectory simulated at 25Hz using the Base parameter set in Table \ref{tab:parameters}.}
    \end{subfigure}

    \caption{Each subfigure depicts the panel of trajectory analysis for a respective trajectory. The \textit{black} line denotes the trajectory; in the x-y coordinate plane as well as the time series for each coordinate. The \textit{blue} line in each figure represents the inferred anchor position described in Section \ref{sec:model}. In subfigure (b) the \textit{orange} line denotes the actual anchor position. After implementing the changepoint algorithm on each trajectory, the right panel in each subfigure shows the state segments for each time series. Each \textit{green} panel denotes a \emph{Motile} segment and each \textit{red} panel denotes a \emph{Stationary} segment. The shaded \textit{gray} panel in subfigure (b) denotes the differentiating overlap of the inferred vs actual segment panels, i.e. the inference gap detailed in Section \ref{sec:classify:inference-gap}. In the segment duration versus speed plot, the \textit{blue} points in both subfigures denote the inferred segments and the \textit{orange} points in subfigure (b) denote the actual simulated segments. \label{fig:intro}}
\end{figure}

\subsection{Summary of work}
In this manuscript, we develop a model and inference protocol for \emph{in vivo} intracellular cargo transport inspired by observed populations of lysosome trajectories. We model motor-cargo movement using a continuous-time Markov process that features switching between \emph{Stationary} and \emph{Motile} periods. The molecular motor is assumed to move in a piecewise deterministic manner, with random velocities that have magnitudes based on the current state and direction that are biased by the state and direction of the preceding segment. The cargo is assumed to rapidly fluctuate about the motor position, and our observations are of the sequence of cargo positions. In contrast to hidden Markov models that have been developed for intracellular transport \cite{RODING2014140,SvenssonEtAl}, we do not assume that motor velocities are drawn from a prescribed finite number of values. Rather, motivated by the observations published in \cite{RayensEtAl}, we assume that the speed distribution is a continuous one. 

In Figure~\ref{fig:intro} we display one experimentally observed particle (i.e. lysosome) trajectory and another trajectory simulated from a parameter set intended to mimic lysosome data. In each case our segmentation algorithm CPLASS \cite{CPLASS-git} has been applied and the inferred ``anchor process'' is shown in blue. Throughout this work we label segments with inferred speed greater than 100 nm/s as \emph{Motile} and label others as \emph{Stationary}. This is a threshold that is useful for visualization but does not play a role in the segmentation algorithm. For the lysosome trajectory, time segments associated with inferred motile states are shaded with light green and inferred stationary states are shaded 
pink. For the simulated data, we know the actual state at every time point and so we can assess the validity of the inferred state. Time steps that are correctly inferred to be Motile are shaded with a green background while correctly identified Stationary time steps are shaded red. The remaining time steps are periods of incorrect state inference. We will refer to the percentage of time spent incorrectly identified as the \emph{inference gap.}

Once segmentation of paths is fixed, it is natural to discuss a population's speed distribution \cite{ENCALADA2011,schroeder2020lps,RayensEtAl}. A typical objective is to compare the speed distribution of transport by two different motors, which requires a robust characterization of the data and quantification of uncertainty. However, as we observe below, studying velocity distributions in terms of segment counts can be very sensitive with respect to the behavior of segmentation algorithms and the experimental frame rate. In Section~\ref{sec:CSA:theory} we discuss a different notion of speed distribution that is robust with respect to segmentation assessment. This alternative characterization focuses on proportion of time spent at or below given speed levels, and to emphasize the distinction, we call this the Cumulative Speed Allocation (CSA). If we view transport speed as a function of time and use the terminology of non-Markovian stochastic process theory, the CSA is the invariant measure of the speed time-series. We can therefore calculate the long-term CSA of our model using Renewal-Reward theory \cite{ciocanel2020renewal}, Section~\ref{sec:CSA:theory}. We subsequently report the performance of the CSA, Section~\ref{sec:CSA:inference}, both in its capacity to robustly differentiate between distinct parameter sets, and in its utility in being able to find a parameter set that replicates \emph{in vivo} data. Moreover, in Section~\ref{sec:MSD}, we contrast the CSA results with reports on the performance of the more standard Mean-Squared Displacement (MSD) analysis of intracellular transport.

\section{Model Development and Numerical Methods\label{sec:model}}

Organelles, vesicles, and other intracellular cargo are often bound to one or more microtubules simultaneously through multiple molecular motors \cite{walcott2022modeling,Klumpp2008}.
Rather than model these cargo-motor-microtubule interactions explicitly, we treat the motor-microtubule connections through a solitary phenomenological \emph{anchor} process.  
We model the anchor as switching between \emph{Stationary} and \emph{Motile} phases. Each of these categories represents a collection of distinct biophysical states, but based on prior observations of \emph{in vivo} transport, we assume they are too numerous to distinguish and treat separately.

We constructed cargo trajectories by first simulating the sequence of anchor states, each of of which consists of a state $J_k$ (1 for \emph{Motile} and 0 for \emph{Stationary}), and a jointly distributed segment velocity $V_k \in \mathbb{R}^2$ and segment duration $\tau_k > 0$, where $k \in \mathbb{Z}_+$. The transition probabilities are described below. With the sequence of anchor segments established, there is a natural translation to a continuous-time Markov chain process $\{J(t), A(t), V(t)\}_{t \geq 0}$. The observed cargo locations are then a sequence of independent Gaussian random variables centered about the anchor positions evaluated over a uniformly discretized sequence of times $\{t_0, t_1, \ldots, t_n\}$. 

When the segment velocities and segment durations are not independent, the continuous-time process $\{J(t), A(t), V(t)\}_{t \geq 0}$ is non-Markovian. This is because the durations will be non-exponential. Consequently, it is necessary to either compute an invariant measure for the process $\{J(t)\}_{t\geq 0}$ or to simulate a burn-in period to allow the system to approach an invariant measure. We opted for the latter throughout this work. Based on Renewal-Reward analysis described in Section~\ref{sec:CSA:theory}, we computed the expected cycle time (the time from the beginning of one \emph{Stationary} state to another) we ran a burn-in time that allows for at least five cycles. We validated the success of this initialization when computing the Cumulative Speed Allocation (see below) and observing that the time spent in the \emph{Stationary} state approximately matched the expected value. 

\subsection{Definition of the anchor segment process}

In our model, we assume that anchor states depend solely on the previous state value. We define the 0-1 stochastic process $\{J_k\}_{k \geq 0}$ to have the following transition probabilities:
\begin{equation} \label{eq:J-transition-matrix}
    \begin{aligned}
    \text{\emph{Stationary}} \rightarrow \text{\emph{Motile}} &: \quad P(J_{k+1}=1|J_k=0) = p;\\
    \text{\emph{Stationary}} \rightarrow \text{\emph{Stationary}} &: \quad P(J_{k+1}=0|J_k=0) = 1-p;\\
    \text{\emph{Motile}} \rightarrow \text{\emph{Stationary}} &: \quad P(J_{k+1}=0 | J_k=1) = q;\\
    \text{\emph{Motile}} \rightarrow \text{\emph{Motile}} &: \quad P(J_{k+1} = 1| J_k = 1) = 1-q.\hspace*{1.5cm}
    \end{aligned}
\end{equation}
We initialized the burn-in period by sampling from the distribution $P(J_0 = 1) = \frac{p}{p+q}$, which is the steady-state distribution of the transition matrix defined by \eqref{eq:J-transition-matrix}. We set $X(0) = (0,0)$ sampled the initial direction from $\text{Unif}(0,2\pi)$. Subsequent directions depend on the previous direction and current state. If the new state is Stationary, then the direction is left unchanged. If the new state is Motile, then there are three possibilities: (1) with probability $\Prev$ the new direction is the reverse of the previous one; (2) with probability $\Pcont$, the anchor continuous in the same direction, but with a different velocity; and (3) with probability $1 - \Prev - \Pcont$, the new direction is chosen from $\text{Unif}(0,2\pi)$ independent of previous state information. The first two options are typical of tug-of-war models for molecular-motor-based transport, assuming that there are velocity changes along a microtubule. The third option stems from motors associated with a cargo switching to a different microtubule in the cytoskeletal network and assuming processive dominance. Each of these outcomes were observed in the data collected for \cite{RayensEtAl} and \cite{RayensEtAl2} and our parameter sets contain qualitative matches to what was observed.

The speed and duration of each of the state segments were defined conditionally on the current state. If the current state was $J_k = 0$, then the speed $S$ was $0$ and the duration $\tau$ was randomly selected from the distribution $\text{Exp}(1/\sigma)$, see Table~\ref{tab:parameters}. If the current state was $J_k = 1$, then the speed $S$ was randomly selected from the distribution $\text{Gamma}(\alpha,\beta)$. The segment durations were either assumed to be independent of the speed, or to have the following structure: 
\begin{displaymath}
    \tau\big|_S \sim \left\{\begin{array}{ll}
    \text{Exp}\big(S/\bar{D}\big), & \text{Dependent Model;} \\
    \text{Exp}\big(E(S)/\bar{D} \big), &
    \text{Independent Model}.
    \end{array} \right.
\end{displaymath}
In the gamma-distributed model we use, $E(S) = \alpha/\beta$. 

The time-dependent anchor process was calculated from the segment chain. Denoting the initial anchor position, $(a_0,b_0)$, we defined the sequence $(a_n, b_n) = (A(t_n), B(t_n)$ by the equations:
\begin{equation} \label{eq:cargosteps}
    a_n = \sum^{n}_{i=1}\nu^x_i (t_i - t_{i-1}); \quad b_n = \sum^{n}_{i=1}\nu^y_i (t_i - t_{i-1}); 
\end{equation}
for $i = 1, \dots, n$, where $t_i$ denotes the time step and $\nu^x_n$ and $\nu^y_n$ denote the anchor speed in the $x$ or $y$ directions at time $t_n$. The cargo observations are then assumed to be independent Gaussian fluctuations around a sequence of anchor locations:
\begin{equation} \label{eq:cargosteps}
    x_i = a_i + \sqrt{\sigmac} \epsilon^{(x)}_i, \quad y_i = b_i + \sqrt{\sigmac} \epsilon^{(y)}_i, 
\end{equation}
where $\{(\epsilon^{(x)}_i, \epsilon^{(y)}_i)\}_{i=1}^n$ are a sequence of independent and identically distributed 2D standard normal random variables with the identity matrix as a covariance. Most simulated paths were generated to have a fixed number of 200 observations.

In this work, we used three parameter sets for our simulations, summarized in Table \ref{tab:parameters}. The Base parameter set is qualitatively similar to parameters inferred from the lysosomal transport data, but with a clear separation from zero in the speed distribution. The Contrast parameter set was selected to make the separation between Stationary and Motile phases even more clear, and to emulate kinesin-1-based transport \emph{in vitro} as described in Jensen et al.~\cite{JensenEtAl}. The Mimic parameter set was selected to match the CSA empirically observed from peripheral lysosomes in the data collected for \cite{RayensEtAl}. Here, the separation between the Stationary and Motile states is not clear. In the Supplementary Information, we display simulated trajectories at various frame rates. 

\begin{table}[hbt!]
\centering
%\begin{tabular}{c l r}
\begin{tabular}{c|c|c|c|l}
\hline
Parameter & Base & Contrast & Mimic & Description
\\ \hline
$n$ & 200 & 200 & 200 & Number of Observations \\
$p$ & 1 & 1 & 1 & Probability \emph{Stationary} to \emph{Motile}\\
$q$ & 0.5 & 0.5 & 0.5 & Probability \emph{Motile} to \emph{Stationary}\\ 
$\alpha$ & 8 & 16 & 0.5  & Speed Shape Parameter \\ %speed_alpha
$\beta$ & 0.02  & 0.02 & 0.005 & Speed Rate Parameter \\ %speed_beta
 $\Bar{D}$ & 300 nm & 900 nm & 200 nm & Average Distance Traveled\\
 $\sigma$ & 5 s & 3 s & 1 s & Average \emph{Stationary} Duration\\
$\Prev$ & 0.3 & 0.3 & 0.3 & Probability of reversal \\
$\Pcont$ & 0.3 & 0.3 & 0.3 & Probability of same direction \\
$\sigmac$ & 0.1 & 0.1 & 0.1 & Noise magnitude of Cargo\\
%$\sigmaa$ & 0 & 0 & 0 & Noise magnitude of Anchor \\
$\Delta$ & \multicolumn{3}{c|}{$\{0.001, 0.01, 0.04, 0.05, 0.1, 1\} \mathrm{s}$} & Time Step\\
$t_f$ & $n \Delta$ & $n \Delta$ & $n \Delta$ & Final simulation time\\
\hline
\end{tabular}
\caption{\textsc{Parameters} "Base" denotes the parameters used to model simulated lysosome trajectories at any frame rate. "Contrast" denotes the parameters used to model kinesin-1 \emph{in vitro} transport. "Mimic" denotes the parameters used to model 20Hz experimental lysosome trajectories. }
\label{tab:parameters}
\end{table}

\subsection{Segmentation Algorithm \label{sec:changept}}

Our analysis proceeds by segmenting each path under the model assumption presented in the previous section. In parallel work, we have developed a segmentation algorithm called CPLASS (Continuous Piecewise-Linear Approximation - Stochastic Search) \cite{CPLASS-wip,CPLASS-git}. Essentially, we assume a linear Gaussian model for the data and perform regression in the space of piecewise-linear functions in two dimensions. For each proposed set of changepoints, we associate to segments the maximum likelihood vector of velocities. To compare two changepoint vectors we have a score function that is equal to the log-likelihood of each proposed changepoint-velocity-vector pair minus a penalty for complexity and for physically implausible speeds. There are two user-selected parameters: a complexity penalization parameter $\beta_{\mathrm{CPLA}}$ and a speed threshold parameter $s_{\mathrm{CPLA}}$. The stochastic search aspect of the algorithm is that we perform a Metropolis-Hastings search for the maximum score, effectively assuming that there is a Gibbs measure on the space of all possible changepoint vectors \cite{CPLASS-wip}. The results presented in this paper are based on selecting the changepoint vector with the highest score that is encountered during 5000 steps of a Metropolis-Hastings walk.

\section{Motor transport summaries and observation frame rates}
\label{sec:results}

The classical method for summarizing the movement of microparticles is mean-squared displacement (MSD). If $\{X_m\}_{m = 1}^M$ is a set of trajectories, with shared observation times $\{t_0, t_1, \ldots\}$ that are uniformly spaced with common duration $\delta = t_{i+1} - t_i$. The paths may be of varying lengths, denoted $\{n_m\}$. Then we write the pathwise MSD
\begin{equation} \label{eq:msd}
    \msd_m(\delta j) = \frac{1}{n_m - j - 1} \sum_{i=0}^{n_M-j} |X(t_{i + j}) - X(t_{i})|^2. 
\end{equation}
To average over the paths at each time point, let $M_j$ be the number of paths that have increments of size $\delta j$ and let $\{m_{i_j}\}$ be an enumeration of these paths. Then we can write the ensemble MSD
\begin{equation} \label{eq:emsd}
    \msd(\delta j) = \frac{1}{M_j} \sum_{m_j} \msd_{m_{i_j}}(\delta j). 
\end{equation}

Mean-squared displacement is a useful tool for describing long-term movement of particles, but does not capture the short-term behaviors that give rise to the long-term transport. This is because, for trajectories that arise from state-switching  mechanisms, the quantity $\msd(\delta j)$ ($j$ small) is averaged over many different movement modalities. As detailed in the Introduction, the use of segmentation algorithms allows for the decomposition of paths into distinct modes of transport. Given this segmentation framework, the question arises whether there is a useful summary statistic, separate from MSD, that is robust in characterizing populations of particles.

In previous work \cite{RayensEtAl,RayensEtAl2}, we have advocated for reporting transport properties in terms of the proportion of time spent in different states. In this work we generalize this principle to a sliding scale of values that we will call cumulative speed allocation (CSA). Adopting the statistical convention of using ``\,$\widehat{\phantom{p}}\,$'' for inferred quantities, suppose that a set of trajectories have been decomposed into (State, Duration, Velocity)-triplets: $\{\hat{J}_k, \hat{\tau}_k,\hat{V}_k\}_{k = 1}^K$. Then we define the inferred CSA to be the function
\begin{displaymath}
    \csahat(s) := \Big(\sum_{k = 1}^K \widehat{\tau}_k \mathbb{1}_{[|\widehat{V}_k|,\infty)}(s) \Big) \Big/ \Big(\sum_{k=1}^K \widehat{\tau}_k\Big).
\end{displaymath}
In other words, for every $s \geq 0$, the CSA is the inferred proportion of time spent at speeds less than or equal to $s$.

In Section~\ref{sec:robustness} we offer some motivation for this choice of statistic. In Section~\ref{sec:classify:inference-gap} and Section~\ref{sec:CSA:theory} we assess its sensitivity to observational frame rate. For contrast, in Section~\ref{sec:MSD}, we present a method for estimating MSD from segmented data, but when looking at the impact of frame rate, we demonstrate a need for improved theory. 

\subsection{Distribution of speeds versus allocation of time}
\label{sec:robustness}

While segmentation allows for a more refined assessment of the mechanistic underpinnings of molecular motor transport, there remains the challenge of finding summary statistics that are (1) robust to measurement methods and (2) sensitive enough to reveal differences between qualitatively distinct populations. One natural mode of analysis is to look at the distribution of inferred speeds \cite{EncaladaEtAl}. However, this turns out to be a surprising example of non-robustness with respect to segmentation algorithm.
\begin{figure}[H]
    \centering
    \includegraphics[width=0.9\linewidth]{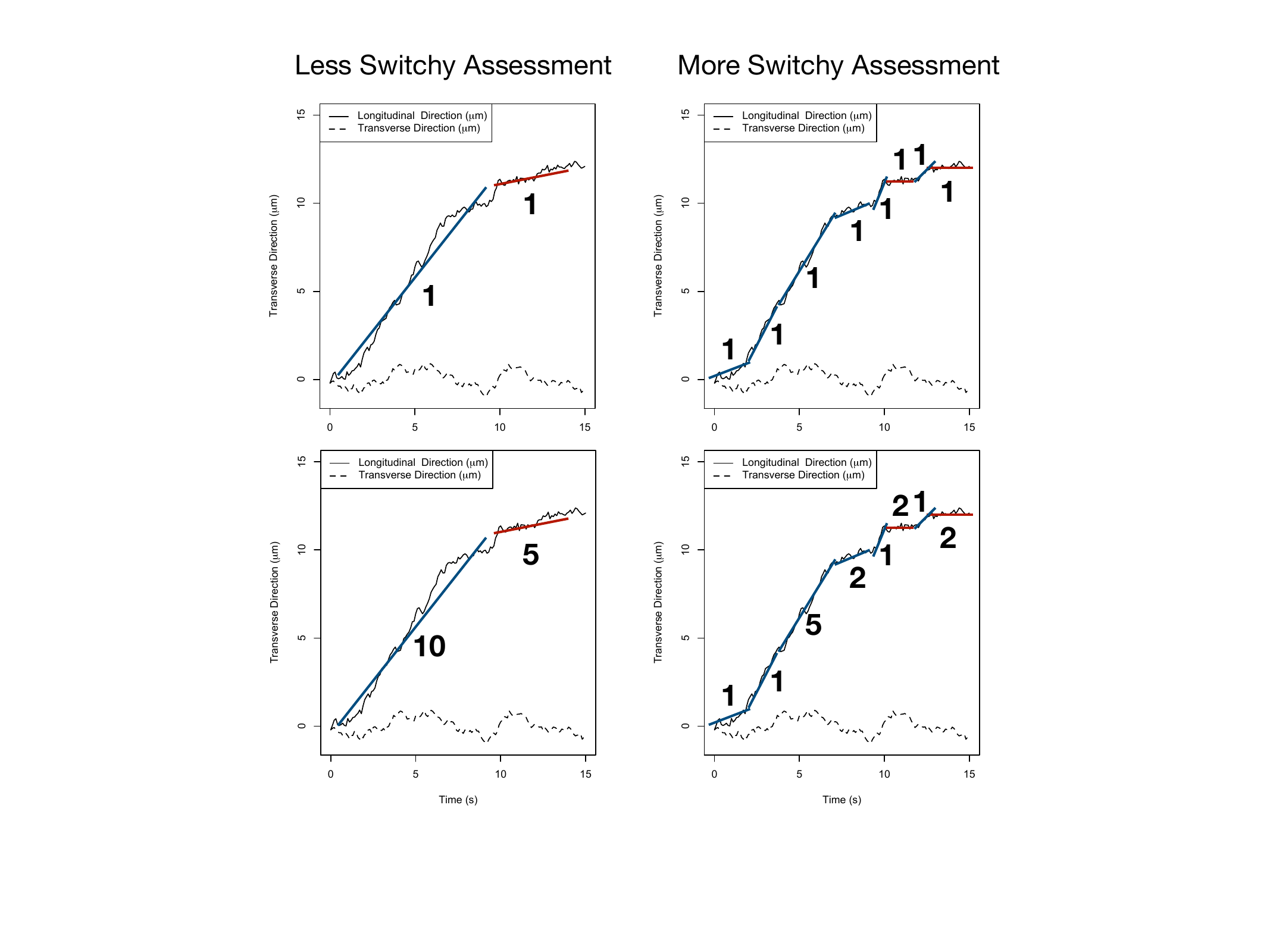}
    \caption{\textsc{Robustness with respect to segmentation.} In this figure we conduct a thought experiment concerning the method of summarizing segment properties. The displayed path is a quantum dot transported along a microtubule \emph{in vitro} \cite{JensenEtAl}. Left and center panels show possible segmentations by algorithms that have two different sensitivity levels. Blue segments are labeled \emph{Motile}, while red segments are labeled \emph{Stationary}. Theoretical summaries of an empirical CDF for speed $\widehat{F}(s)$ and the inferred CSA $\csahat(s)$ are presented in the table at right.}
    \label{fig:robustness}
\end{figure}

Consider the following thought experiment, depicted in Figure~\ref{fig:robustness}. An observed microparticle moves at roughly 1 micron per second for ten seconds and then is \emph{Stationary} for the next five seconds. A highly sensitive segmentation algorithm might decompose the path into numerous short segments, while a less sensitive algorithm might correctly assess the path as having two state segments. Define the empirical speed cumulative distribution function (CDF) as follows:
\begin{equation}
    \widehat{F}(s) := \frac{1}{K} \sum_{k = 1}^K \mathbb{1}_{[|\widehat{V}_k|,\infty)}(s). 
\end{equation}
From a modeling point of view, this is the eCDF of the distribution of anchor speeds assuming all states are Motile. We see in our thought experiment that the ratio of segments with speeds above 0.5 $\micron$/s to those below 0.5 $\micron$/s changes from 3:1 to 1:1 when we shift from a sensitive ``switchy'' assessment to a ``less switchy'' one. By contrast ratio of time spent above 0.5 $\micron$/s to time spent below 0.5 $\micron$/s is roughly 2:1 in both cases. 
\begin{figure}[h!]
    \centering
    \includegraphics[width=0.45\textwidth]{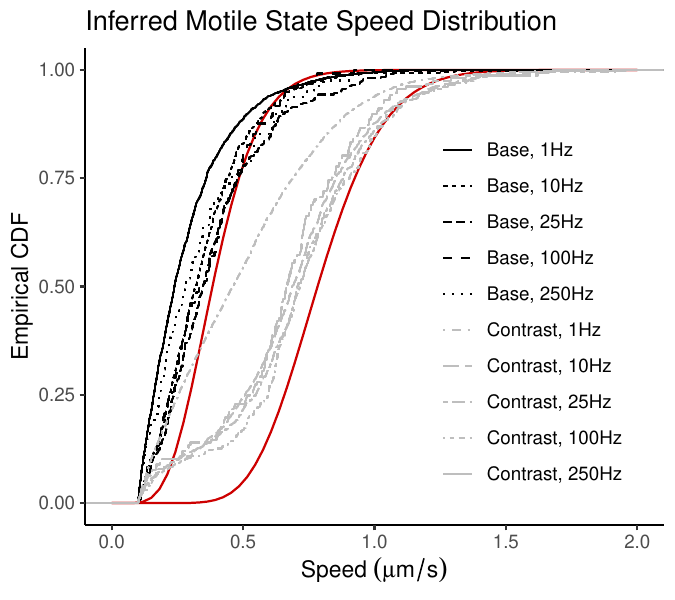}    
    \includegraphics[width=0.4\textwidth]{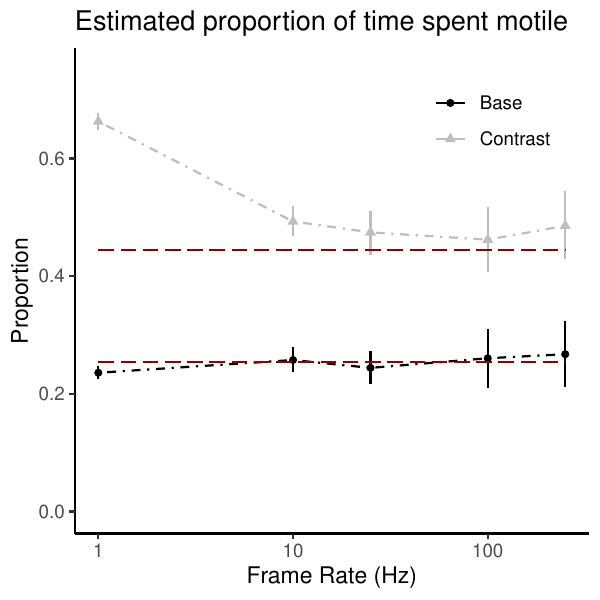}~
    \caption{\textsc{Simulated data, speed counts vs time allocation.} \textbf{(Left)} Empirical CDFs for segment speeds during inferred \emph{Motile} states. The red curves indicate the CDF used during simulations for two distinct parameter sets. Substantial disparities exist between the true and inferred CDFs for both parameter regimes across all frame rates. \textbf{(Right)} Estimates for proportion of time spent \emph{Motile} across multiple frame rates using a 100 nm/s for a \emph{Stationary}/\emph{Motile} threshold. We see our first instance a bias/variance trade-off in selecting frame rates.}
    \label{fig:count_vs_time}
\end{figure}

The difference in the reporting methods can also be observed in light of varying frame rates. Let $\widehat{F}_M(s)$ denote the empirical CDF of segment speeds conditioned on the event that the speed is greater than 0.1 $\micron$/s. We can think of this as the speed distribution of the \emph{Motile} state. In Figure~\ref{fig:count_vs_time} we display 
$\widehat{F}_M(s)$ for simulated ensembles of two parameter sets (see Table~\ref{tab:parameters}), observed at five different frame rates. The true \emph{Motile} state speed distributions are given in red. We see that there is a substantial gap regardless of the frame rate. By contrast, in the right panel, we display estimates for time spent above 0.1 $\micron$/s for the same simulated data sets and the same segmentation by CPLASS. For frame rates faster than 10 Hz, the true proportion falls within 95\% bootstrap confidence interval. 

We do note a first occurrence of a bias-variance tradeoff in the segmentation analysis though. There exists substantial bias in the slow frame rate toward overestimation of time spent in the Motile state. Meanwhile, at faster frame rates, the bootstrap confidence interval (computed by resampling paths with replacement) is consistently larger. We explore the causes of each in the next section.

\subsubsection{Quantifying the Inference Gap}
\label{sec:classify:inference-gap}

For simulated data, we can directly assess what proportion of the time the inference protocol yields mislabeled states. We call the percentage of mislabeled time steps the \emph{Inference Gap}. To be precise, from the inferred segment triplets $(\widehat{J}_k, \widehat{\tau}_k, \widehat{V}_k)$ we can assign a triplet to every observation time $\big(\widehat{J}(t_i), \widehat{\tau}(t_i), \widehat{V}(t_i)\big)$. If we set $s_*$ be a threshold speed above which a segment is labeled \emph{Motile}, then we can define the Inference Gap of a single path to be 
\begin{displaymath}
    \text{Inference gap (Pathwise)} := \frac{100}{n} \sum_{i = 1}^n 1\{\widehat{J}(t_i) = J(t_i)\}.
\end{displaymath}
\begin{figure}[h!]
    \centering
    \includegraphics[width=0.34\textwidth]{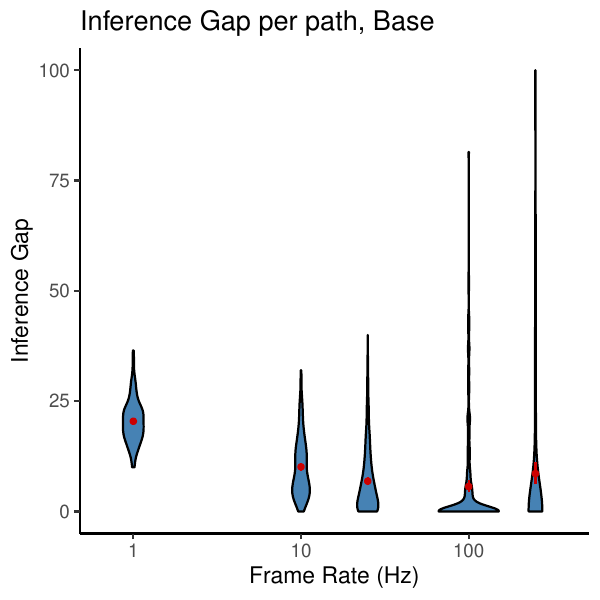}~
    \includegraphics[width=0.3\textwidth,
    trim={1.2cm 0 0 0},clip]{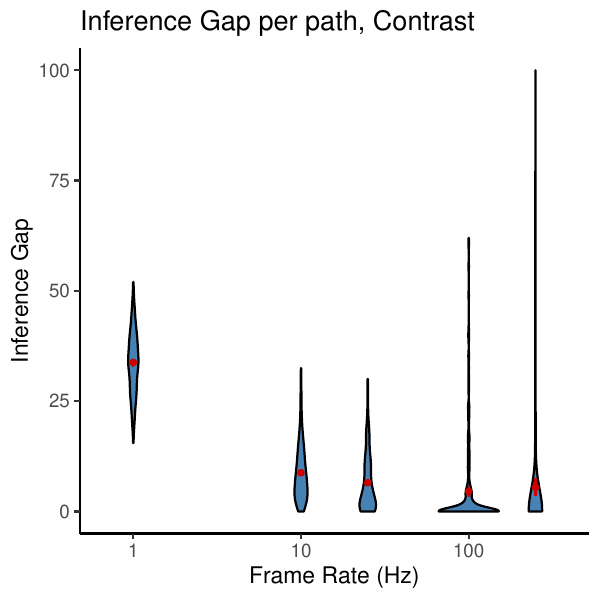}~
    \includegraphics[width=0.34\textwidth]{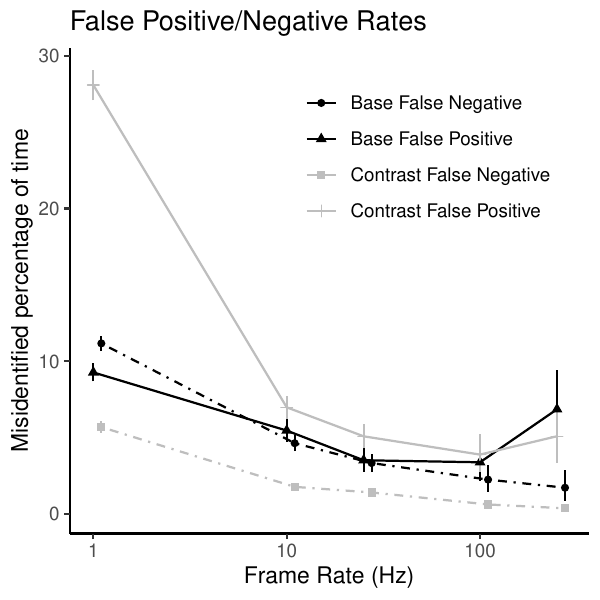}    
    \caption{\textsc{Inference Gaps} Simulated data sets: Base (\textbf{Left}) and Contrast (\textbf{Center}). The path-by-path distribution of the percentage of time spent misidentified in the \emph{Stationary}/\emph{Motile} dichotomy with speed threshold $s_* = 0.1 \micron$/s. \textbf{(Right)} Ensemble averages for the percentage of mislabeled time steps, broken down in terms of False Positives and False Negatives. Error bars computed from bootstrap resampling of paths. }
    \label{fig:inf-gap}
\end{figure}

In Figure~\ref{fig:inf-gap}, we display a summary of the segmentation/classification protocol applied to simulations of the base parameters set observed at different frame rates. For both the Base and Contrast data sets, the average pathwise Inference Gap (red dots) decreases as the frame rate increases from 1 Hz to 100Hz. For both parameter sets, the average inference gap is larger at 250 Hz than at 100 Hz, indicating some optimal frame rate in between (if minimal inference gap is the objective). 

Consistent with the right panel of Figure~\ref{fig:count_vs_time}, the variance of the pathwise Inference Gaps is increasing. This is because of the fixed number of observations across all frame rates. Since individual segments account for larger proportions of the particle paths, a single mislabeled segment can result in a path being mislabeled for a large percentage of its time. For example, in an ensemble of paths observed with a 250 Hz frame rate, it is typical for at least one path to be entirely mislabeled. This happens when a path consisted of just one segment, and that segment has a velocity close to the threshold. A typical error in velocity estimation can lead to mislabeling, which would result in an Inference Gap of 100 for that path. On the other hand, at 25 Hz and 100 Hz, the vast majority of mislabeling mistakes are brief and the average Inference Gap is at its lowest.

The Inference Gap statistic can be seen as the sum of the two canonical types of statistical error: \emph{false positive time points}, which the segmentation-and-classification algorithm has labeled \emph{Motile} when they are truly \emph{Stationary}, and \emph{false negative time points}, which have been labeled \emph{Stationary} when they truly \emph{Motile}. This breakdown is informative when we see phenomena like the Inference Gap being twice as large for the Contrast parameter set compared to the Base parameter set, when both are observed at 1 Hz. In the right-hand panel of Figure~\ref{fig:inf-gap}, we see that the culprit is a high rate of false positives. False positive time can occur is there is a short Motile segment that interrupts a period of stationarity. The displacement that occurs can cause an entire one-second time step to be labeled Motile, when in fact the majority of the time step was Stationary. There is also a natural way for false negatives to occur: through reversals. During a one-second time step, a motor might be \emph{Motile} in one direction half of the time step, and then \emph{Motile} in the other direction for the other half-step. The total displacement would be negligible, inferred as a period of stationarity, but in truth the motor was \emph{Motile} the full time step. 

A close inspection of false positive and false negative rates reveals a spectrum of ways that segmentation/classification algorithms will inevitably mislabel motor behavior. For real motors, there may be incorrect labeling because the velocities of segments are not as constant as the model would present them to be. A piecewise constant model would fail to capture continuously varying speeds. This is particularly an issue for small frame rates (large time steps) when displacements are emerging from a sum of (possibly) varying velocities. On the other hand,  high frame rate observations can lead to the algorithm ``chasing noise'' instead of true state signals. This is due to the relative size of fluctuations we compared to the smaller true displacement that occurs during time steps. This leads to false positive time points where 1) the segmentation algorithm associates the large amounts of noise with significant changes between too many consecutive short segments or 2) segments with large jumps in the noise between the  changepoints are averaged together. 

\subsection{Cumulative Speed Allocation (CSA, $\Psi(s)$): Theory \label{sec:CSA:theory}}

All of the analysis in the preceding section was presented for a speed threshold $s_* = 0.1 \micron$/s. But much of the same logic applies for other choices of the threshold. Rather than pick a handful of thresholds and consider them individually, the Cumulative Speed Allocation (CSA) essentially makes the \emph{Stationary}/\emph{Motile} assessment across all choices of $s_*$ simultaneously. But more than being just a generalization of \emph{Stationary}/\emph{Motile} classification, the application of Renewal-Reward theory \cite{ciocanel2020renewal} permits a predictive analysis of the CSA for a given parameter set. In this section we present our main mathematical result, which is the theoretical CSA for trajectories generated by the model described in Section~\ref{sec:model}.

As a reminder, we write trajectory in terms of a sequence of anchor states $\{(J_i, \tau_i, V_i)\}$ and the associated continuous-time anchor position, velocity, and state: $\{J(t), A(t), V(t)\}_{t \geq 0}$, where $A$ and $V$ are understood to be vectors at each time point. For every speed $s \geq 0$, we define the theoretical CSA as follows:
\begin{equation}
    \Psi(s) = \lim_{t \to \infty} P(|V(t)| < s).
\end{equation}
In other words, it is the long-term (or, informally, ``steady-state'') probability that a motor is in transport with a speed less than a given value $s$. By ergodicity, we can equivalently define the CSA in terms of a long-term running average:
\begin{equation}
    \Psi(s) = \lim_{T \to \infty} \frac{1}{T} \int_0^T \!\! \mathbbm{1}_{|V(t)| < s} \d t.
\end{equation}
This time-average perspective amenable to analysis via Renewal-Reward theory. The key insight is to decompose the trajectory into regenerative cycles that are independent and identically distributed, and then apply a functional central limit theorem \cite{ciocanel2020renewal,serfozo2009basics}. 

\begin{theorem}[Cumulative Speed Allocation]
\label{thm:csa}
    Let a sequence of state-duration-velocity triplets be defined according to the model framework of Section~\ref{sec:model}, and let $\{J(t), A(t), V(t)\}$ be the associated continuous-time state-position-velocity triplet. Define 
    \begin{equation}
        \rho = \frac{\sigma q \tilde{\alpha}}{\beta \bar{D}}
    \end{equation}
    where
    \begin{equation}
        \tilde{\alpha} = \left\{ \begin{array}{cl}
           \displaystyle \alpha - 1,  & \textrm{Dependent segment speed and duration;} \\
           \displaystyle \alpha  & \textrm{Independent segment speed and duration.}
        \end{array} \right.
    \end{equation}
    Then
    \begin{equation}
        \Psi(s) = \frac{\Gamma_c(s \with \tilde \alpha, \beta) + \rho}{1 + \rho}.
    \end{equation}
    where $\Gamma_c(\cdot \with \alpha, \beta)$ is the CDF of the anchor speed distribution, which is assumed to be $\text{Gamma}(\alpha, \beta)$. 
\end{theorem}
\begin{proof}
If we define the random variable $\Delta T$ to represent the random time between regenerations and $\Delta I_s$ to be the time spent during a cycle with speed less than $s$, then the theory presented in \cite{ciocanel2020renewal} implies that
\begin{equation}
\Psi(s) = \frac{E(\Delta I_s)}{E(\Delta T)}.
\end{equation}

We first consider $E(\Delta T)$, where $\Delta T$ is the regeneration time for the moments when we enter the \emph{Stationary} state. That is to say, $\Delta T = \tau_1 + \tau_2$ if the sequence of states is \emph{Stationary}, \emph{Motile}, \emph{Stationary} and $\Delta T = \tau_1 + \tau_2 + \tau_3$ if the sequence of states is \emph{Stationary}, \emph{Motile}, \emph{Motile}, \emph{Stationary}, etc. Recall from Table~\ref{tab:parameters} that $\sigma$ is duration of a \emph{Stationary} state, and let $\bar \tau$ denote the expected duration of a \emph{Motile} state. Since the number of consecutive \emph{Motile} states before returning to \emph{Stationary} is geometrically distributed with a success probability $q$ (recall Table \ref{tab:parameters}) the expected regeneration duration is 
\begin{displaymath}
    E(\Delta T) = \sigma + \frac{\bar \tau}{q}
\end{displaymath}
where $\sigma = E(\tau_i | J_i = 0)$ and $\bar \tau = E(\tau_i | J_i = 1)$. The expectation of the duration depends on which model assumption we are using. If $\tau$ depends on the segment speed $S$, then (assuming $\alpha > 1$ and recalling from Table~\ref{tab:parameters} that $\bar D$ is the average distance traveled per segment), we have that
\begin{equation}
    \Bar{\tau} = \int_{0}^{\infty}E(\tau|s)p(s)\d s = \int_{0}^{\infty} \frac{\bar D}{s} \frac{\beta^\alpha}{\Gamma(\alpha)}s^{\alpha-1}e^{-\beta\nu}\d s = \frac{\beta \bar D}{\alpha-1},
\end{equation}
If $\tau$ is independent of the speed, 
\begin{displaymath}
    \Bar{\tau} = \frac{\beta \Bar{D}}{\alpha}.
\end{displaymath}
 
This shows that if we choose a motor's speed to have a shape parameter $\alpha \leq 1$, then we need to choose the duration independent of speed. Otherwise, the motor would spend essentially all of its time (asymptotically) in near-\emph{Stationary} states (\emph{Motile} with very small velocity).

Now, $\Delta I_s$, the time spent during a cycle with a speed less than $s$ $\mu \text{m}/\text{sec}$, can be written 
\begin{displaymath}
  \Delta I_s = \tau_1 + \sum \tau_i \mathbbm{1}{ \{ |V_i| \leq s\}}.
\end{displaymath} 
Taking expectations, we must compute $E\big(\tau_i \mathbbm{1}{ \{ \nu_i \leq \nu^*\}}\big)$ under the two assumptions concerning the dependence of segment duration on speed.

Let $E_j(X) = E(X \given J=j)$ denote that we are taking the expected value of a segment quantity conditioned on the state of the particle during that segment, recalling that $J = 0$ indicates a \emph{Stationary} motor and $J = 1$ indicates a \emph{Motile} motor.

If $\tau$ is independent of the speed $|V|$, then 
\begin{displaymath}
    \begin{aligned}
        E_1(\tau \mathbbm{1}{ \{ |V| < s\}}) &= E_1(\tau) \, E_1(\mathbbm{1} \{ |V| \leq s\}) \\
        &= \frac{\beta \Bar{D}}{\alpha} P_1 (|V| \leq s) \\
        &= \frac{\beta \Bar{D}}{\alpha}\int_0^{s}\frac{\beta^\alpha}{\Gamma(\alpha)}\nu^{\alpha -1}e^{-\beta \nu}d\nu \\
        &= \frac{\beta \Bar{D}}{\alpha - 1} \Gamma_c(s \with \alpha, \beta).
    \end{aligned}
\end{displaymath}

If $\tau$ depends on the speed, then
\begin{displaymath}
\begin{aligned}
    E_1(\tau \mathbbm{1}{ \{ |V| < s\}}) &= E_1(\tau \mathbbm{1}{ \{ |V| < s\}}) \\
    &= \int_0^\infty E_1(\tau \mathbbm{1}{ \{ |V| < s\}} \given |V| = \nu) \gamma(\nu \with \alpha, \beta) d\nu \\
    &= \int^{\infty}_0 E_1(\tau \given |V| = \nu)\frac{\beta^\alpha}{\Gamma(\alpha)}\nu^{\alpha-1}e^{-\beta \nu}d\nu \\
    &= \int^{s}_0 \frac{\Bar{D}}{\nu} \frac{\beta^\alpha}{\Gamma(\alpha)}\nu^{\alpha - 1} e^{-\beta \nu}d\nu \\
    &= \Bar{D} \beta \frac{\Gamma(\alpha - 1)}{\Gamma(\alpha)} \int^{s}_0\frac{\beta^{\alpha - 1}}{\Gamma(\alpha - 1)}\nu^{(\alpha + 1)- 1} e^{-\beta \nu}d\nu \\
    &= \frac{\bar{D} \beta}{\alpha} \gamma(s \with \alpha + 1, \beta)
\end{aligned}
\end{displaymath}

\end{proof}

\subsection{Cumulative Speed Allocation (CSA, $\Psi(s)$): Inference \label{sec:CSA:inference}}

With the main CSA theorem in hand, we can assess the quality of this summary statistic from multiple perspectives on performance: (1) whether we can infer model parameters from an experimental data set and then use the parameters to create simulations that faithfully reproduce CSA curves; (2) assessing the quality of uncertainty quantification in service of distinguishing between ensembles generated from different parameter sets, and (3) assessing robustness with respect to frame rate of observation.
\begin{figure}[h]
    \centering
    \includegraphics[width=0.75\linewidth]{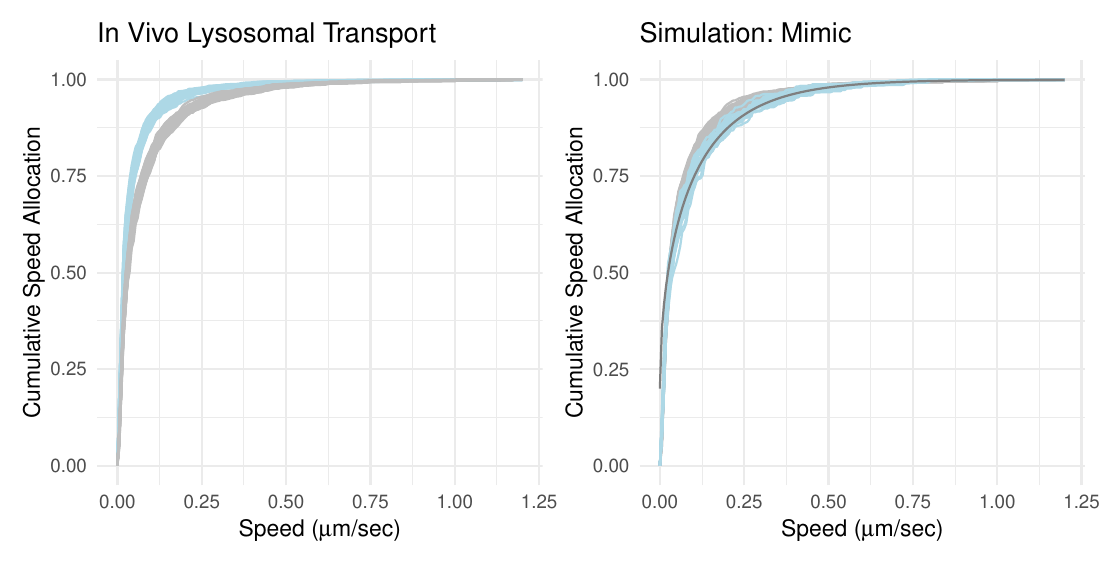}
    \caption{\textsc{Cumulative Speed Allocation} \textbf{(Left.)} CSA curves inferred for lysosomal transport in the perinuclear region (light blue) and periphery (gray) of cells observed by Rayens et al \cite{RayensEtAl}. The significance of the difference is illustrated by the distribution of CSA curves computed for bootstrap resamples of the respective trajectory ensembles. \textbf{(Right)} Theoretical CSA, and CSA curves of bootstrap resampled trajectory ensembles, for simulated motor-cargo complexes generated from the Mimic parameter set and observed at 20 Hz (matching the lysosomal data). Gray curves are the target distribution of CSA curves, corresponding to lysosomes in the periphery of observed cells.}
    \label{fig:CSAData}
\end{figure}

\subsubsection{Toward parameter inference and faithful simulation}
The first of these concerns is addressed briefly in Figure~\ref{fig:CSAData} and will be more addressed more systematically elsewhere. In the left panel of Figure~\ref{fig:CSAData}, we display the result of applying the CPLASS segmentation algorithm to 200 trajectories randomly chosen from the perinuclear region and periphery of cells studied by Rayens et al.~\cite{RayensEtAl}. Each member of the CSA curve ensembles -- (light blue for the perinuclear region and gray for the cell periphery) -- is the inferred CSA calculated from a bootstrap resampling of the 200 paths. The clear gap between the CSA ensembles shows that the CSA is able to distinguish paths in the two regions, and their primary difference is in the proportion of time that their lysosomes are moving at speeds 0.5 $\micron$/s or slower. This is consistent with the conclusions of this and the follow-up effort \cite{RayensEtAl2} that lyosomes near the endoplasmic reticulum spend significantly more time in \emph{Stationary} or near-\emph{Stationary} states.

In order for a summary statistic to move beyond simply being descriptive, it must be possible to build a stochastic model that will generate simulated ensembles for which the application of CPLASS will result in similar CSA curves as the target data. In Figure~\ref{fig:CSAData} we display the results of a proof-of-concept numerical experiment. Using \emph{ad hoc} methods, we found the collection of parameters in the Mimic parameter set (Table~\ref{tab:parameters}) tracked reasonably with peripheral lysosomes from the Rayens et al.~study. The dark gray curve shows is the theoretical CSA for the Mimic parameter set, calculated using Theorem~\ref{thm:csa}, and the overlap of the bootstrapped CSA ensemble (light blue) shows that these paths are not significantly different from the periphery lysosomes from a CSA-perspective. (As we note also in the Discussion, there are a host of other properties that we do not attempt to match here. In particular, we leave model selection and precise inference of the joint distribution of consecutive segment velocity directions for future work.)

\subsubsection{Grappling with small-speed statistical artifacts}
In the preceding section, we analyzed what we called the Base and Contrast parameter sets. These are distinguished from the Mimic data set in that there is a clear separation between \emph{Stationary} and \emph{Motile} states. This is common, for example, when studying \emph{in vitro} data, particularly for cargo transported by kinesin-family motors. The separation between \emph{Stationary} and \emph{Motile} states is in evidence in Figure~\ref{fig:CSABaseContrast}. Both the theoretical CSA curves (dark gray) and the bootstrapped ensembles of CSA curves (orange for Base, and yellow for Contrast) are essentially flat for small speeds. This means that \emph{Motile} segments that are less than 0.25 $\micron$/s (Base) or 0.5 $\micron$/s (Contrast) rarely appear. The gray curves showing the bootstrapped CSA-ensemble for peripheral lysosomes shows no such absence of slow \emph{Motile} segments.
\begin{figure}[H]
    \centering
    \includegraphics[width=0.85\linewidth]{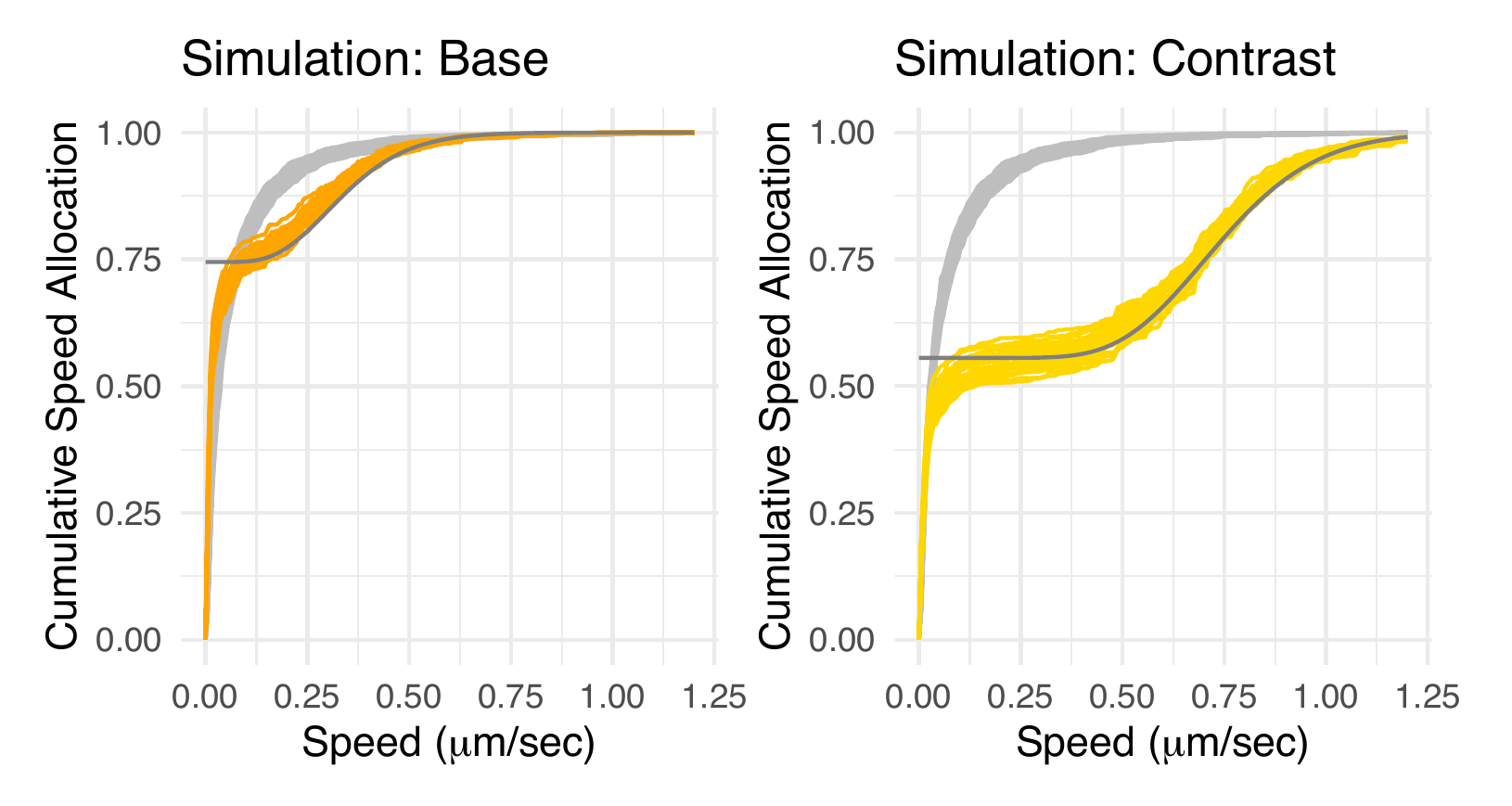}
    \caption{\textsc{Base and Contrast simulations}. CSA computation for a set of 250 simulated trajectories each of the parameter sets in Table \ref{tab:parameters}, observed at 25 Hz. The \textit{dark gray} line denotes the theoretical CSA of for each parameter set and the \textit{gray} curves are the computed CSA for experimental lysosome trajectories in the periphery region of the cell from \cite{RayensEtAl}. The \textit{orange} and \textit{yellow} curves are CSA curves for bootstrapped subsamples of the respective path ensembles.}
    \label{fig:CSABaseContrast}
\end{figure}
We do note that there is an overlap of all CSA curves for very small velocities. This is a statistical artifact that will arise from most segmentation algorithms, and is important to understand when developing theoretical models. The CPLASS algorithm returns a set of inferred segments and assigns a maximum likelihood (MLE) speed to each of them. The MLE speed is never zero, so even if a segment has exactly zero motor velocity, the observational noise will produce a non-zero MLE speed. Our focus on these two parameter sets was motivated in part by a desire to understand the difference between the theoretical CSA and the inferred CSA for motors with clear \emph{Motile}/\emph{Stationary} state separation. The clear difference between the simulated sets and the lysosomal CSA curves indicates that the distribution of speeds \emph{in vivo} genuinely fill out a full speed distribution from very slow to roughly 0.6 $\micron$/s.

\subsubsection{Bias vs.~variance again}
In Figure~\ref{fig:CSASims}, we investigate one more instance of the Bias-Variance tradeoff that has been evident throughout our study. Using the Base parameter set we simulated paths to be observed at 1 Hz, 25 Hz and 250 Hz, with 250 trajectories in each ensemble. To emphasize the degree of variation, for each frame rate, we created 100 subsample-ensembles consisting of 50 paths each. For each subsample, we calculated a CSA given the exact motor (anchor) position at all times (blue curves), and a CSA based on particle locations and the application of the CPLASS algorithm (green curves). 
\begin{figure}[H]
    \centering
    \includegraphics[width=1\linewidth]{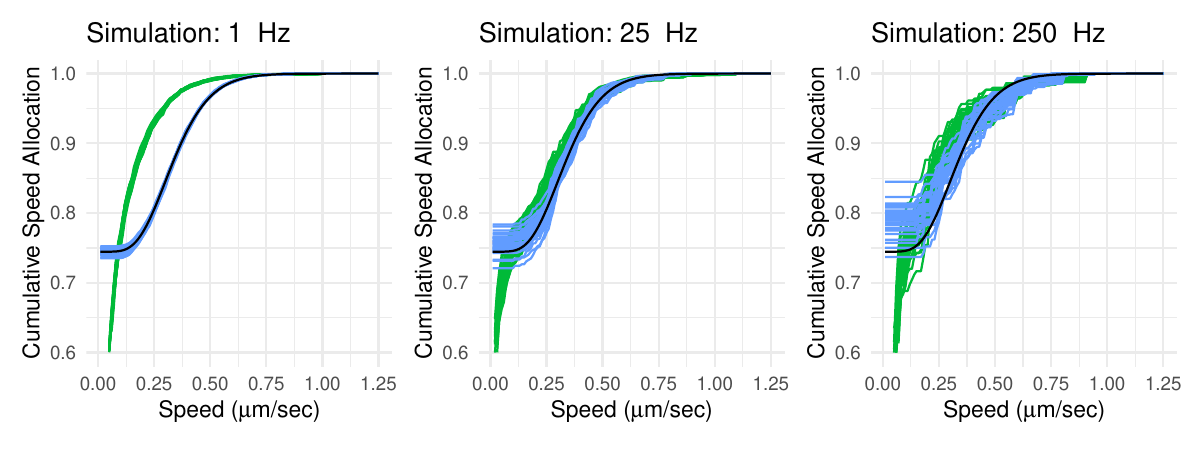}
    \caption{\textsc{CSA Bias and Variance.} CSA ensembles generated from different Base parameter set in Table \ref{tab:parameters}. We simulated three sets of trajectories at the frame rates of 1 Hz, 25 Hz, and 250 Hz. The \emph{black} line in each plot denotes the theoretical CSA computed using the derivation in Section \ref{sec:CSA:theory}. The \emph{green} lines in each plot denote the inferred CSA computations for each set of simulated trajectories. The \emph{blue} lines denote CSA bootstrap samples of the true simulated trajectories.}
    \label{fig:CSASims}
\end{figure}

Calculating CSA curves from the true motor position allows us to distinguish two sources of error in CSA estimation. First, due to the fixed number of observations across different frame rates, there is less parameter-related information contained in high frame rate date. Paths observed at 250Hz contain relatively few switches, which yields fewer speed observations and less information about average segment durations. Note, moving right to left in Figure~\ref{fig:CSASims}, the blue curves become more concentrated around the true CSA as paths become longer. This is because the anchors explore a full and representative set of behaviors in the time provided during the data set. However, observation of the anchor behavior is compromised at 1 Hz. For reasons detailed in Section~\ref{sec:classify:inference-gap}, there is a substantial bias toward underestimating speeds in 1 Hz data, and the CSA is biased toward incorrectly large values. On the other hand, inference on 100Hz and 150Hz data is very good. The CSA curves inferred from paths closely mirror the information provided by the anchor positions. The variation that is present is due to a lack of ground truth to work with. This is the root of the the bias-variance tradeoff first observed for \emph{Stationary}/\emph{Motile} classification in Figure~\ref{fig:count_vs_time}.

\subsection{Mean-Squared Displacement (MSD)}
\label{sec:MSD}

The analysis would not be complete without addressing the implications of choosing different frame rates on the evaluation of MSD. While we do not advocate for using MSD as a general tool, it remains an effective method for communicating information about the ``transport scale'' of intracellular cargo movement. For paths that switch between stationary and ballistic states, the MSD curve should display a transition: for small lags, the MSD is effectively constant, dominated by stationary segments that have constant variance; and then there should be a turnover to a linear phase with a slope that equals the effective diffusivity, an averaging of the two states \cite{pavliotis2014stochastic}. 

\begin{figure*}[h]
    \centering
    \begin{subfigure}[t]{0.35\textwidth}
        \centering
        \includegraphics[width=1\linewidth]{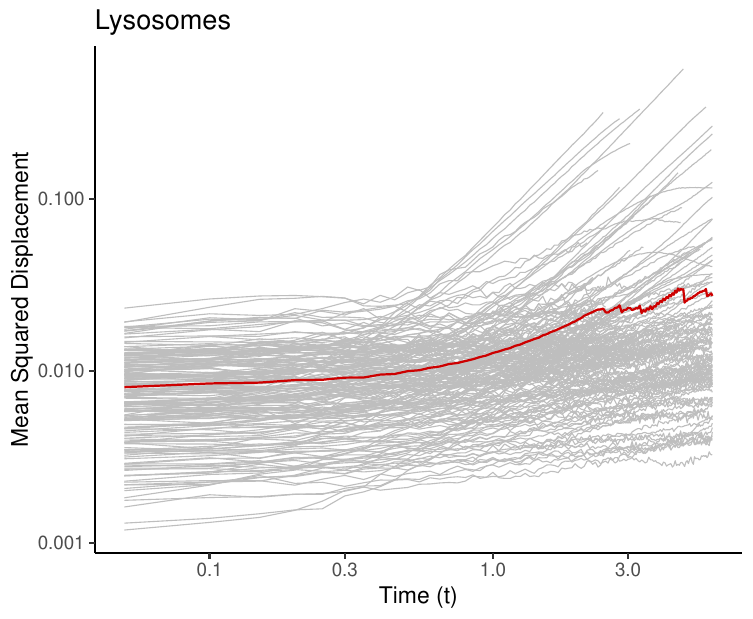}
        \caption{Experimental MSD}
    \end{subfigure}%
    ~     
   \begin{subfigure}[t]{0.35\textwidth}
        \centering
        \includegraphics[width=1\linewidth]{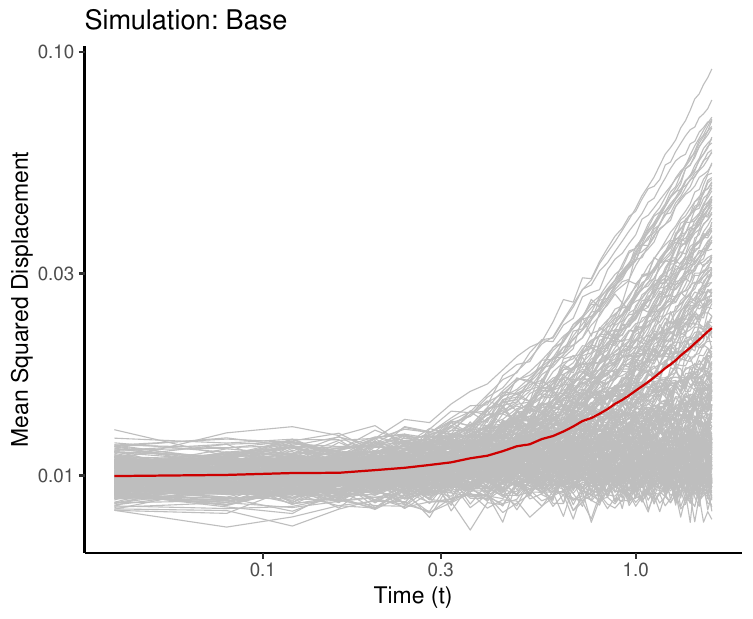}
        \caption{25Hz MSD}
    \end{subfigure}%
    
    \begin{subfigure}[t]{0.35\textwidth}
        \centering
        \includegraphics[width=1\linewidth]{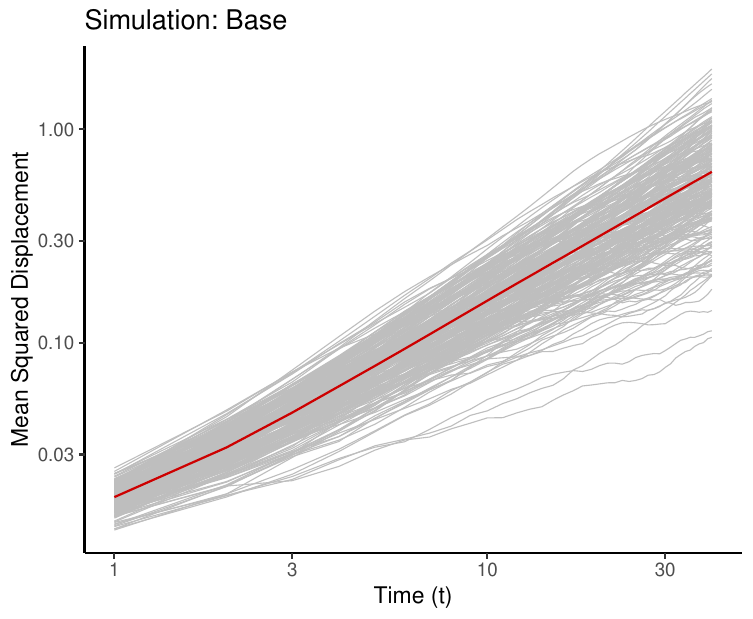}
        \caption{1Hz MSD}
    \end{subfigure}%
    ~
    \begin{subfigure}[t]{0.35\textwidth}
        \centering
        \includegraphics[width=1\linewidth]{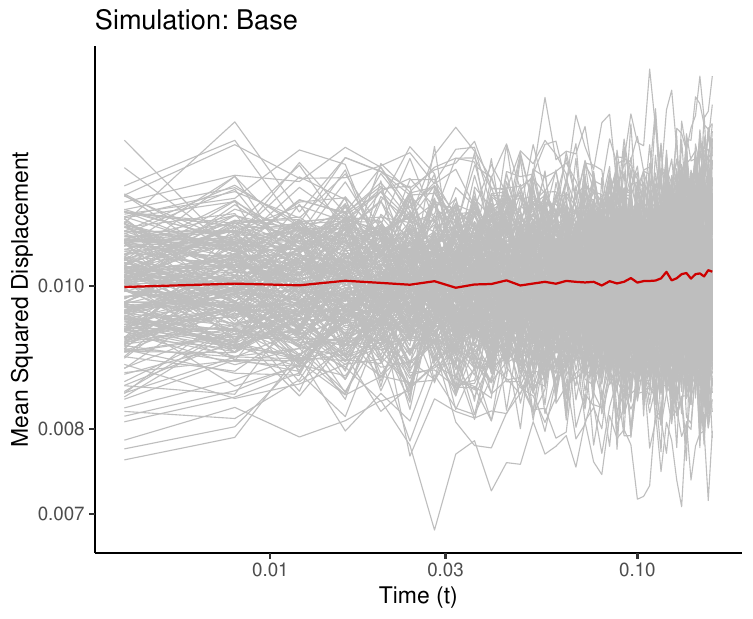}
        \caption{250 Hz MSD }
    \end{subfigure}
    \caption{The pathwise MSD of 250 simulated trajectories using the Base parameter set in Table \ref{tab:parameters} for frame rates 1 Hz, 25 Hz, and 250 Hz. Subfigure (a) contains the MSD analysis of the experimental lysosome trajectories in the periphery from \cite{RayensEtAl}. The \textit{gray} curves denotes the pathwise MSDs and the \textit{red} curve denotes the ensemble average MSD.
    \label{fig:msd}}
\end{figure*}

From a modeling perspective, there are multiple approaches to computing effective diffusivity \cite{brooks1999probabilistic,hughes2011matrix,krishnan2011renewal,BressloffNewby,popovic2011stochastic,ciocanel2020renewal}. But many of these concepts do not translate into methods for statistical inference. Consider, for example, the ensembles of MSDs in Figure \ref{fig:msd} that show the intrinsic difficulty in estimating effective diffusivity by way of pathwise MSDs. Depending on the frame rate, a transition to the final slope may not be evident. Only in the 1 Hz simulation data do we see the asymptotically linear state. However, the individual path MSDs are noisy, and there are ambiguities in what range of time lags to use for a linear fit.

\begin{figure}[h!]
    \centering
    \includegraphics[width=0.6\linewidth]{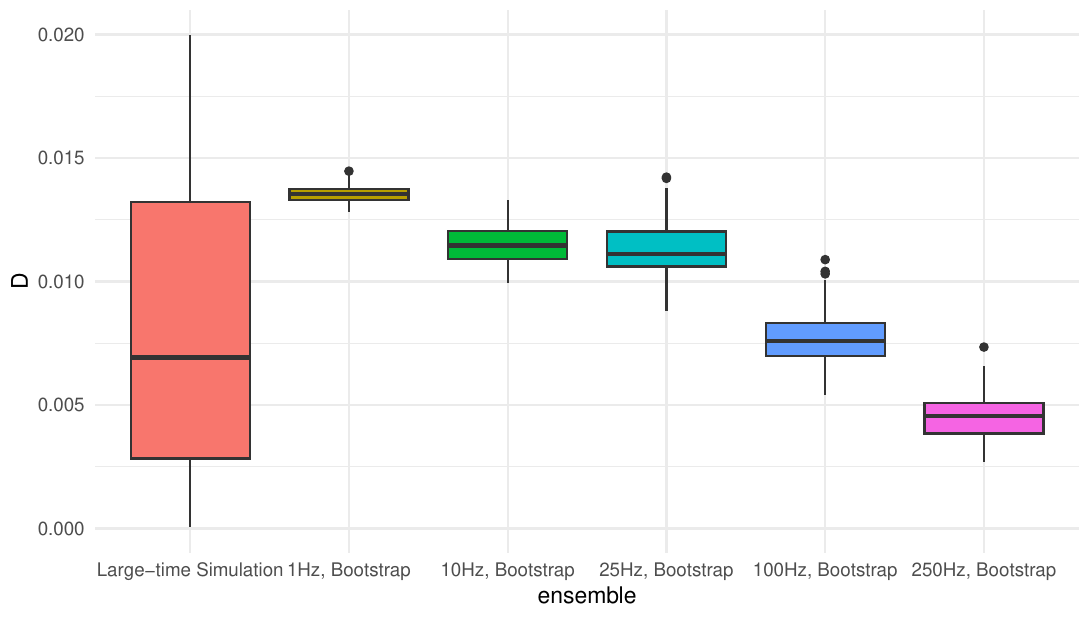}
    \caption{\textsc{Diffusivity Comparisons.} Estimated diffusivity of 250 simulated trajectories using the Base parameter set in Table \ref{tab:parameters} at frame rates of 1 Hz, 10 Hz, 25 Hz, 100 Hz, and 250 Hz. The box plots show the variance in the estimated diffustivities of the trajectories simulated at various frame rates. The Large-time Simulations estimate was generated by simulating 250 anchor-position paths for 10,000 seconds and computed the squared displacement divided by the total path duration. The takeaway is that MSD is a noisy statistic with poorly understood biases that can lead to inaccurate characterization of data.}
    \label{fig:DiffBox}
\end{figure}

The Renewal-Reward framework, introduced for intracellular transport by Krishnan et al.~\cite{krishnan2011renewal} and implicitly used by Popovic et al.~\cite{popovic2011stochastic}, and further developed by Ciocanel et al.~\cite{ciocanel2020renewal}, presents an opportunity for a statistical inference method that can be applied to segmented paths. If we assume that the directions of consecutive segments are independent (which is not actually the case for our simulated models), then Ciocanel et al.~\cite{ciocanel2020renewal} showed that the asymptotic diffusivity can be written 
\begin{equation} \label{eq:diffusivity-effective}
    D_{\text{eff}} = \frac{\langle (\Delta X)^2 \rangle}{4 \langle \Delta T \rangle}. 
\end{equation}
The $\Delta X$ and $\Delta T$ averages are computed directly from segment information, and this formula was used in \cite{RayensEtAl} to estimate the effective diffusivity of lysosomes in different regions of a cell. While there is some agreement between this method and the more traditional approach employed in \cite{Payne2014}, we can use our simulated system to show that estimation of diffusivity can be affected by the choice of frame rate. In \cite{Payne2014}, the frame rate was 1 Hz, which yielded MSD slopes that were amenable to linear fit, while in \cite{RayensEtAl}, the frame rate was 20 Hz. 

In Figure \ref{fig:DiffBox}, we show estimates for the effective diffusivity at different frame rates using Equation~\ref{eq:diffusivity-effective} and bootstrap resampling of paths. Notably, the bootstrap confidence intervals mostly fail to overlap. Establishing the ``true'' effective diffusivity is more difficult to establish than one might think. The correlations in the model between consecutive angles make the approximation \eqref{eq:diffusivity-effective} inaccurate. We elected to use a simulation approach and found this method to be unreliable as well. We simulated 250 paths for 10,000 seconds and calculated the mean-squared displacement for each. The variation in the calculated values is displayed in the box on the left.

\section{Discussion}
We have developed a protocol for analyzing and simulating noisy, continuous, and piecewise-linear microparticle trajectories. This movement paradigm is widely observed in single particle tracking experiments for molecular motor transport of intracellular cargo, both \emph{in vivo} and \emph{in vitro}. The modeling and analysis flows through the decomposition of paths into distinct segments of constant velocity with uniform iid Gaussian noise. Ensemble behavior is reported in terms of an estimate for the CDF of the invariant distribution of the (non-Markovian) motor speed process. We call this statistical summary the Cumulative Speed Allocation (CSA). In terms of modeling and faithful simulation, we show that the CSA contains enough information to distinguish between qualitatively distinct motor-cargo populations, and allows for customization of simulations via experimentally-tuned parameters to mimic qualitative different movement regimes. Moreover, we have demonstrated a general need to emphasize inference methods that are robust with respect to choice of frame rate and segmentation algorithm.

Throughout this work, we have highlighted a surprising bias-vs-variance tradeoff that occurs across different frame rates. For example, we observed a bias in the estimation of time spent \emph{Motile} for slow frame rates, but a general increase in the uncertainty for faster frame rates. We investigated this phenomenon by comparing the true state of simulated segments to inferred states and were able to quantify the causes of the inference gaps found at each frame rate. Similarly, estimates for effective diffusivity depended on frame rate in an unexpected way, without a clear answer for what frame rate yields the best estimate. For these reasons, we believe the notion of ``optimal'' frame rate bears further investigation, noting that it ultimately depends on the statistical characterization of interest.

The work we have presented here is fundamentally about the compromise of  collecting ever more precise data concerning molecular motor transport, but encountering fundamental physical limits. Any \emph{in vivo} observation technique risks disrupting the process it seeks to observe and there are hard physical limits to the amount of information that can be gained about profoundly small objects being observed over extremely brief windows of time. There are on-going technology improvements that will affect assumptions underpinning this work, but we believe that tradeoffs of the kind described here are intrinsic to all micro- or nano-particle tracking.

There are many ways in which the model we consider can be modified and improved. For example, we do not thoroughly address the issue of correlation/anti-correlation between the velocities of consecutive segments. Our model includes a probability of reversal and a probability of pausing but then continuing in the same direction, but a systematic study of how to parameterize velocities and directions remains unfinished. This is because the inference of switch rates is profoundly affected by the tendencies of different segmentation algorithms. A ``switchy'' algorithm will break up single motile runs, for example, resulting in an inference of shorter run lengths, but higher correlation between consecutive segments. We conducted an preliminary study to this effect, which informed our parameter set choices, but the fit is qualitative in this respect. 

Another unresolved modeling issue relates to whether we should consider the speed and duration of segments to be dependent. There is a purely physical rationale for why duration could scale inversely speed: when microtubules are short or the microtubule network is dense, faster-moving motor-cargo complexes will have shorter run lengths \cite{WijeratneEtAl,KanekoEtAl,ChewEtAl,walcott2022modeling}. Qualitatively, this inverse relationship between segment speed and duration was observed for the lysosome trajectories studied in \cite{RayensEtAl}, Figure S8. However it is not clear how much of the relationship is due to biophysical considerations and how much is merely a statistical artifact of segmentation algorithms. The theory we have presented in this work addresses both cases, but data that explicitly shows the microtubule network underlying motor transport may help clarify the segment speed-duration relationship for different motors in various environments. Of course, robust joint cargo-motor-microtubule observation data would invite the use of even more realism in simulations, \cite{walcott2022modeling}. The CSA provides a new tool for quantifying qualitative changes in transport due to environment factors and a more complete inference/simulation protocol could inform predictive analysis of how motor transport might respond to changes in the supporting microtubule network.

Finally, further consideration should be given to the observation process itself. There is strong evidence that fast frame rates can lead to noisier inference of particle locations \cite{GUO2023107009,MANLEY2010109,MacGillavry,Kashchuk}. In this study we summarized both the cargo fluctuations about the anchor and the measurement error by a single parameter $\sigma$. We expect that the determination of an ``optimal'' frame rate will be affected by frame-rate dependent noise.

Taken all together, the work we have presented  calls for continued simultaneous modeling of biophysical processes \emph{and} explicit detail for the techniques used to observe them. If the goal of the applied mathematician is to provide insight across scales, then it is necessary for models to address the challenges that arise in accurately assessing essential micro-scale details.

\section*{Use of AI tools declaration}
The authors declare they have not used Artificial Intelligence (AI) tools in the creation of this article.

\section*{Acknowledgments}

This research was supported by the NSF-Simons Southeast Center for Mathematics and Biology (SCMB) through NSF-DMS1764406 and Simons Foundation-SFARI 594594.

\section*{Conflict of interest}

The authors declare there is no conflict of interest.

\bibliographystyle{plain}
\bibliography{SimulationPaper}

\newpage

\section*{Supplementary}

\subsection{Various Trajectories and their respective Analyses}

% SI Figures
% \begin{enumerate}
%     \item Real path (periphery) with inferred states.
%     \item Mimic path with actual and inferred states (caption explains the ``actual state'' color scheme).
%     \item Stacked Base and Contrast paths (with actual and inferred) 1Hz. (Much shorter caption)
%     \item Stacked Base and Contrast paths (with actual and inferred) 25Hz.
%     \item Stacked Base and Contrast paths (with actual and inferred) 250Hz.
% \end{enumerate}
\begin{figure*}[h]
    \centering
        \includegraphics[height=4.0in]{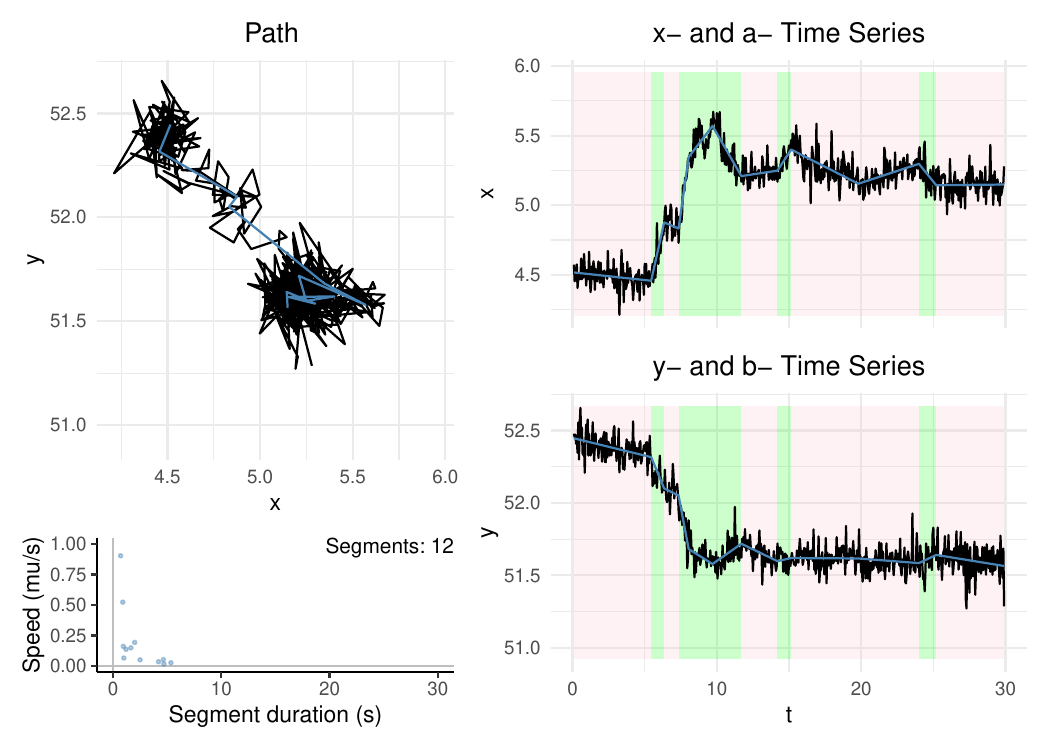}
        
    \caption{Trajectory analysis of the experimental lysosome trajectory in the periphery of a cell from \cite{RayensEtAl}. The figure depicts the panel of trajectory analysis for a respective trajectory. The \textit{black} line denotes the trajectory; in the x-y coordinate plane as well as the time series for each coordinate. The \textit{blue} line in each figure represents the inferred anchor position described in Section 2.2 in the main text. After implementing the changepoint algorithm on each trajectory, the right panel shows the state segments for each time series. Each \textit{green} panel denotes a Motile segment and each \textit{red} panel denotes a Stationary segment. The segment duration versus speed plot shows the speed of each respective inferred segment of the trajectory. Keep in mind that there exist Motile to Motile state changes and Stationary to Stationary state changes that result in side-by-side green panels and red panels.}
    
\end{figure*}
    
\begin{figure*}[h]
    \centering    
   \includegraphics[height=4.0in]{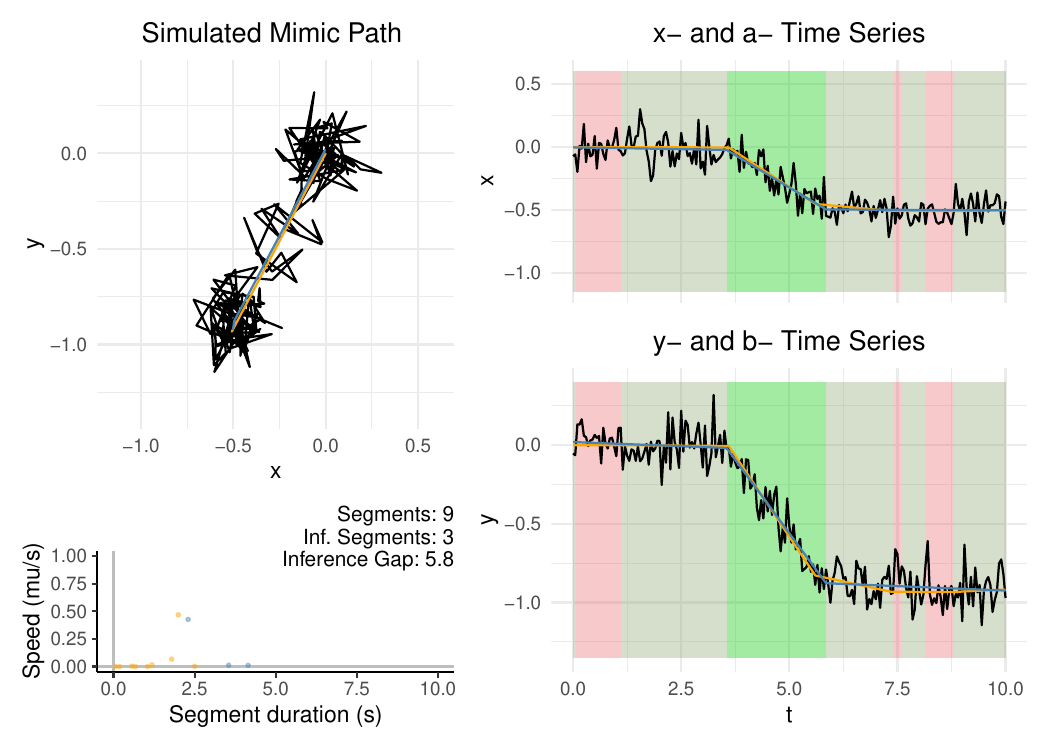}
       
    \caption{A trajectory simulated at 20Hz using the Mimic parameter set in Table 1 in the main text. The figure depicts the panel of trajectory analysis for a respective trajectory. The \textit{black} line and the \textit{blue} line are as defined in Figure 1. The \textit{orange} line in each panel represents the actual (simulated) anchor position described in Section 2.2 in the main text. After implementing the changepoint algorithm on each trajectory, the right panel in each subfigure shows the inferred states overlapping the actual states. Overlapping \textit{green} and \textit{red} panels become \textit{gray}, denoting the inference gap (Section 3.1.1. in the main text), or the difference between the actual simulated segments and those determined by the changepoint algorithm.
    }
\end{figure*}

\begin{figure*}[h]
    \centering
            \begin{subfigure}[t]{0.8\textwidth}
        \centering
        \includegraphics[width=\textwidth]{{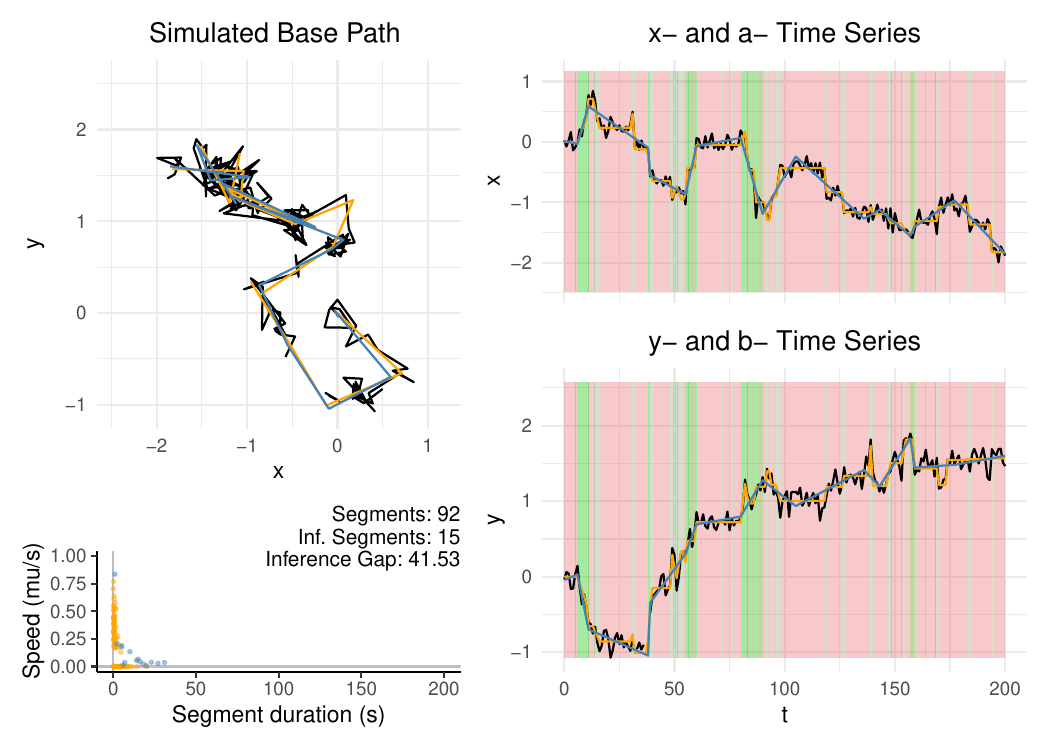}}
    \end{subfigure}
    
    \begin{subfigure}[t]{0.8\textwidth}
        \centering
        \includegraphics[width=\textwidth]{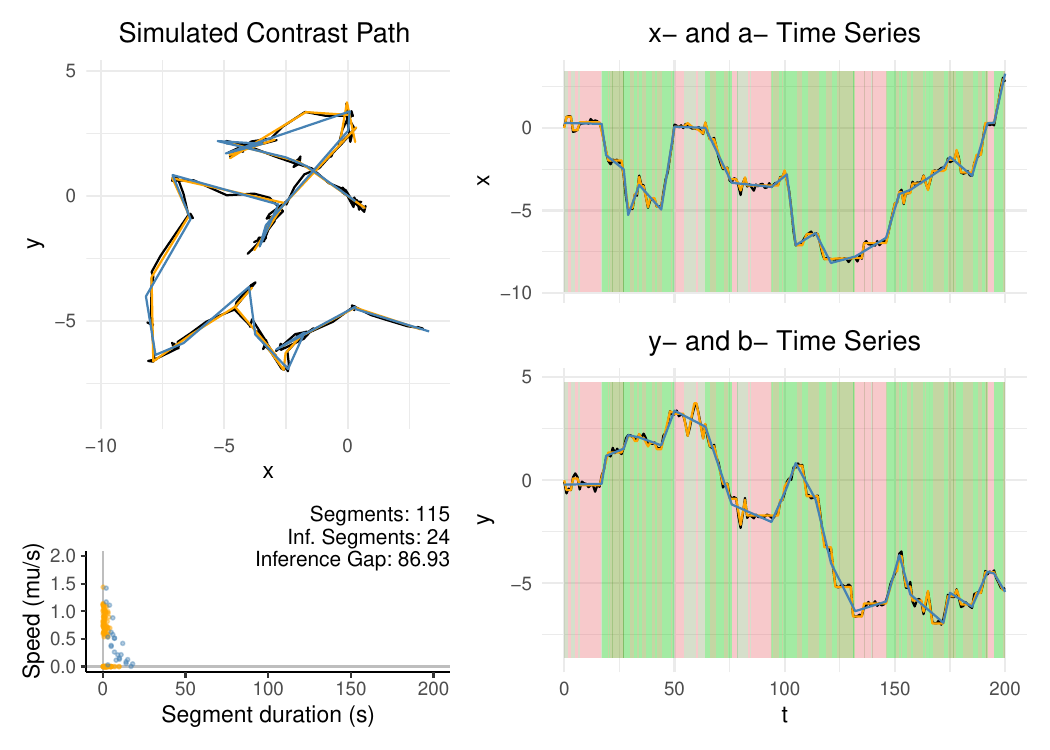}
    \end{subfigure}
        
        \caption{A Base and Contrast trajectory simulated at 1Hz using the base and contrast parameter sets in Table 1 of the main text. The figure depicts the actual and inferred trajectory analysis. }
\end{figure*}

\begin{figure*}[h]
    \centering
            \begin{subfigure}[t]{0.8\textwidth}
        \centering
        \includegraphics[width=\textwidth]{{FigureDraft2/Sim13_CKP_25Hz_cplass_featured_ActInfpath36.pdf}}
    \end{subfigure}
    
    \begin{subfigure}[t]{0.8\textwidth}
        \centering
        \includegraphics[width=\textwidth]{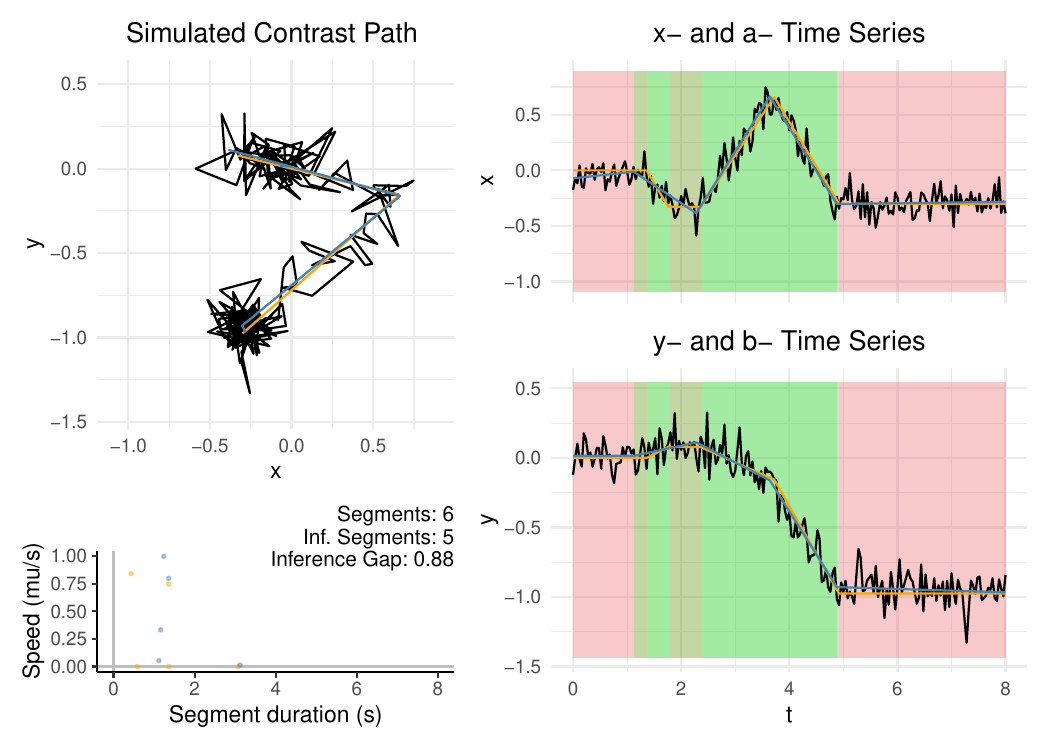}
    \end{subfigure}
        
        \caption{A Base and Contrast trajectory simulated at 25Hz using the base and contrast parameter sets in Table 1 of the main text. The figure depicts the actual and inferred trajectory analysis. }
\end{figure*}

\begin{figure*}[h]
    \centering
            \begin{subfigure}[t]{0.8\textwidth}
        \centering
        \includegraphics[width=\textwidth]{{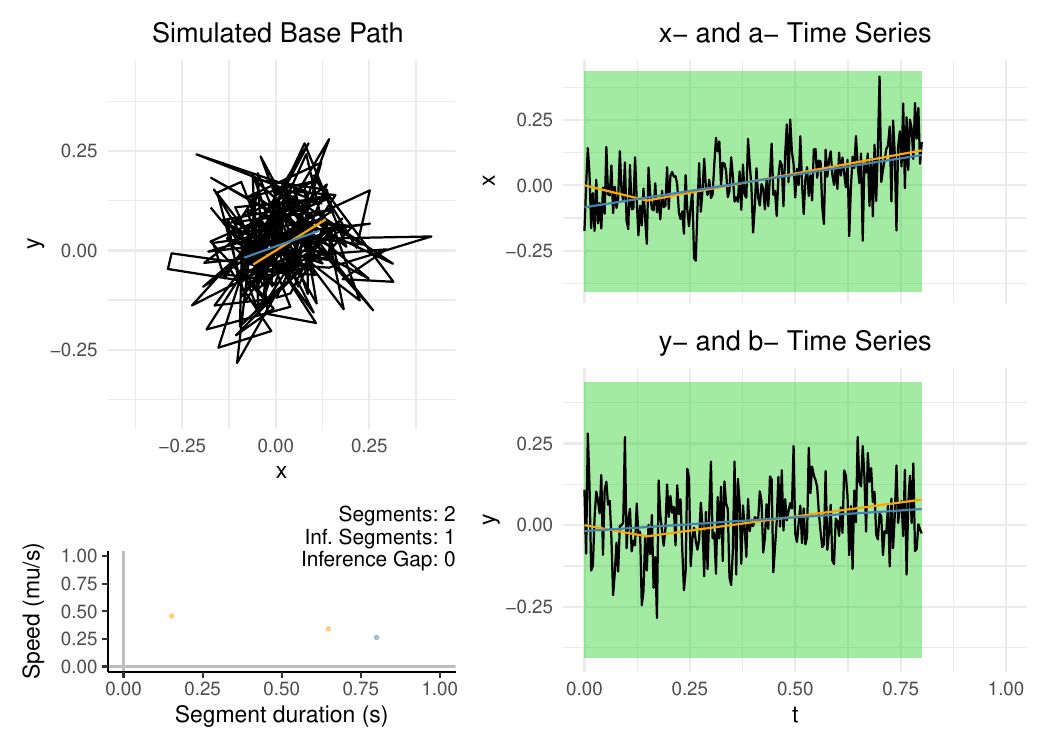}}
    \end{subfigure}
    
    \begin{subfigure}[t]{0.8\textwidth}
        \centering
        \includegraphics[width=\textwidth]{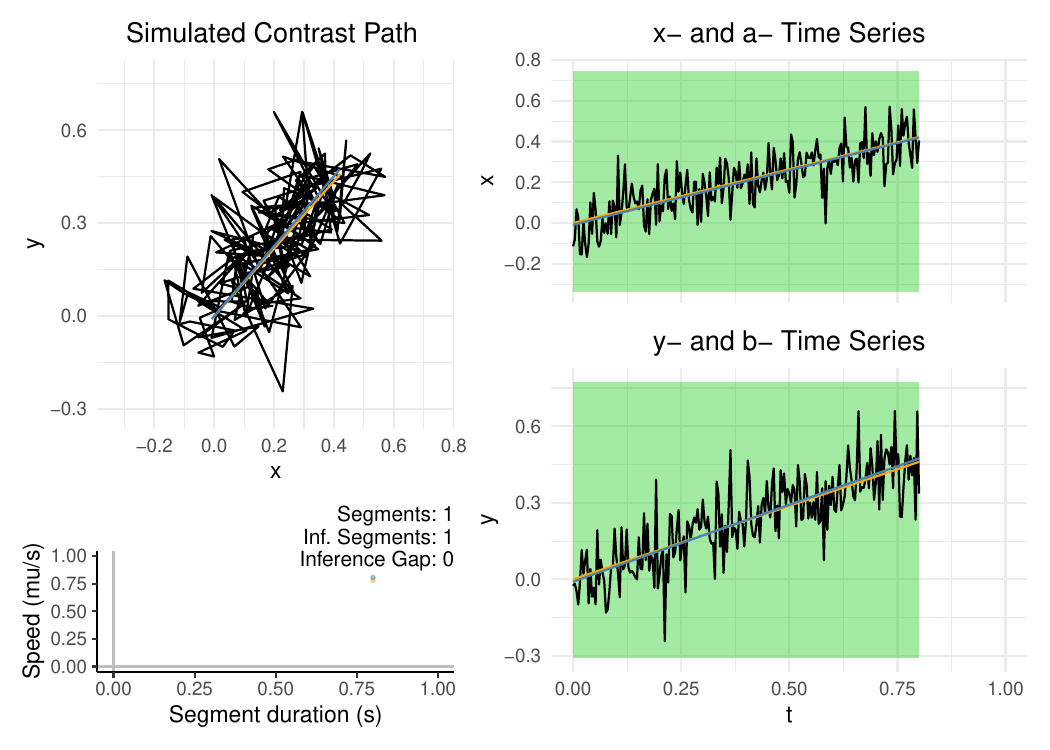}
    \end{subfigure}
        
        \caption{A Base and Contrast trajectory simulated at 250Hz using the base and contrast parameter sets in Table 1 of the main text. The figure depicts the actual and inferred trajectory analysis. }
\end{figure*}

\end{document}